\documentclass[11pt]{article}

\usepackage{latexsym}
\usepackage{graphicx}
\usepackage{amssymb,eepic,epic,latexsym}
\usepackage{amsthm,amsmath}

\newtheorem{definition}{Definition}
\newtheorem{lemma}{Lemma}

\newtheorem{theorem}{Theorem}

\newtheorem{claim}{Claim}

\begin{document}

\title{A New Upper Bound for the VC-Dimension of Visibility Regions\footnote{An extended abstract of this paper appeared at SoCG '11 \cite{gk-nubvcvr-11}.}}

\author{Alexander Gilbers\\
Institute of Computer Science\\
          University of Bonn\\
           53113 Bonn, Germany \\
           gilbers@cs.uni-bonn.de
       \and Rolf Klein\\
       Institute of Computer Science\\
          University of Bonn\\
           53113 Bonn, Germany \\
           rolf.klein@uni-bonn.de
       }

\maketitle

\begin{abstract}
In this paper we are proving the following fact.
Let $P$ be an arbitrary simple polygon, and let $S$ be an arbitrary set of 15 points inside~$P$.
Then there exists a subset $T$ of $S$ that is not ``visually discernible'', that is, 
$T \not=\mbox{vis}(v) \cap S$ holds for the visibility regions $\mbox{vis}(v)$ of all points $v$ in $P$.
In other words, the VC-dimension $d$  of visibility regions in a simple polygon cannot exceed $14$.
Since Valtr~\cite{v-ggwps-98} proved in 1998 that $d \in [6,23]$ holds, 
no progress has been made on this bound.
By $\epsilon$-net theorems our reduction immediately implies a smaller upper bound to the number of guards
needed to cover $P$.
\end{abstract}


\section{Introduction}  %

Visibility is among the central issues in computational geometry, see, e.~g., 
Asano et al.~\cite{ags-vip-00}, Ghosh~\cite{g-vap-07} and Urrutia~\cite{u-agip-00}.
Many problems involve visibility inside simple polygons,
among them the famous art gallery problem: Given a simple polygon~$P$, find a minimum set of 
guards 
whose visibility regions together cover $P$; see O'Rourke~\cite{r-agta-87}. 

In this paper we study a visibility problem that is related to the art gallery problem, and interesting
in its own right. Given a simple polygon $P$ and a finite set $S$ of points in $P$, we call a subset $T$
of $S$ {\em  discernible} if there exists a point $v \in P$ such that $T=\mbox{vis}(v) \cap S$ holds.
In general, one cannot expect all subsets of a given point set in a given polygon to be discernible.
If all subsets of a given point set $S$ are discernible we say that $S$ is \textit{shattered}.

Let us call a number $m$ {\em realizable} if there exists a simple polygon $P$, and a set $S$ of $m$
points in $P$, such that all subsets of $S$ are discernible. If $m \geq 1$ is realizable, so is $m-1$.
The example in Figure~\ref{VierPunkte-fig} shows that~4 is realizable. 

\begin{figure}[hbtp]%
  \begin{center}%
    \includegraphics[scale=1.2,keepaspectratio]{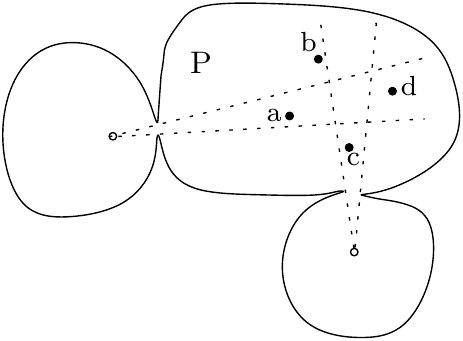}%
    \caption{All subsets consisting of elements that form a contiguous substring of $abcd$ can be discerned from the 
    lower cave of $P$, all others from the left. Hence, 4~is realizable.}%
    \label{VierPunkte-fig}
  \end{center}%
\end{figure}

The biggest realizable number $d$ is called
the {\em VC-dimension} of visibility regions in simple polygons. 
Valtr~\cite{v-ggwps-98} showed that $d \in [6,23]$, and these were the best bounds on $d$
known until today. 
In this paper we show that 15~is not realizable, which implies $d \in [6,14]$.

\begin{theorem}   \label{statement-theo}
For the VC-dimension $d$ of visibility  regions in simple polygons, $d \leq 14$ holds.
\end{theorem}

The classic $\epsilon$-net theorem implies that $O(d\cdot r \log r)$ many stationary guards
with $360^\circ$ degree view are sufficient to cover $P$, provided that each point in $P$ sees
at least an $1/r$th part of $P$'s area. For sufficiently large $r$ the constant hidden in $O$ is very close to~1; 
see Kalai and Matou\v sek~\cite{km-ggwep-97} and Koml{\'o}s et al.~\cite{kpw-atben-92}. Decreasing the upper bound on the VC-dimension $d$ directly leads to more interesting 
upper bounds on the number of guards. 
For a textbook treatment of VC-dimension we refer the reader to Matou\v sek~\cite{m-ldg-02}.

\section{Related Work}         \label{related-sec}%

The VC-dimension of range spaces of visibility regions was first considered by Kalai and Matou\v sek \cite{km-ggwep-97}.
They showed that the VC-dimension of visibility regions of a simply connected gallery (i.e. a compact set in the plane) is finite.
In their proof they start with assuming that a large set (of size about $10^{12}$) of points $A$ inside a gallery is shattered by the visibility regions of the points of a set $B$. They then derive a configuration as in Figure \ref{KalaiMatousekProof-fig}. Here, points $a$ and $b$ should not see each other but the segment $\overline{ab}$ is encircled by visibility segments, a contradiction. This kind of argument plays an important role in our proof of the new bound.
\begin{figure}[hbtp]%
  \begin{center}%
    \includegraphics[scale=0.7,keepaspectratio]{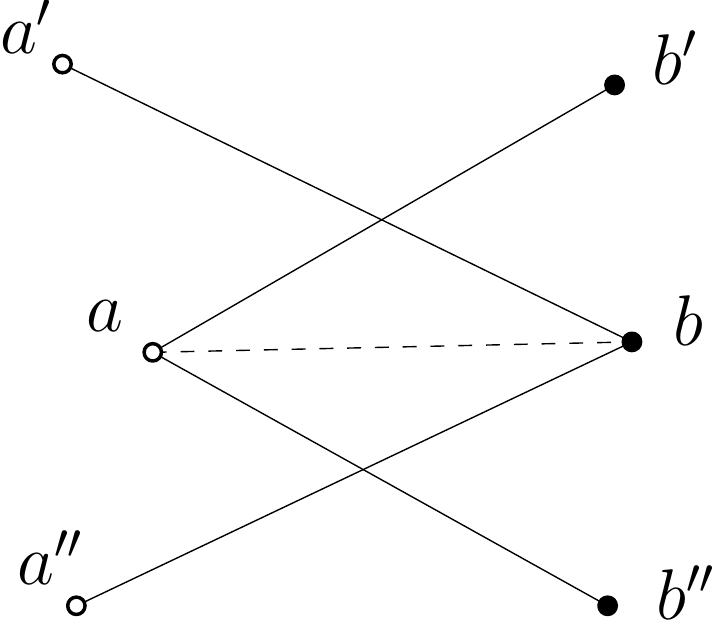}%
    \caption{Segment $\overline{ab}$ is encircled by visibility segments.}%
    \label{KalaiMatousekProof-fig}
  \end{center}%
\end{figure}
They also gave an example of a gallery with VC-dimension~$5$. Furthermore, they showed that there is no constant that bounds the VC-dimension
if the gallery has got holes.  
For simple polygons, Valtr~\cite{v-ggwps-98} gave an example of a gallery with VC-dimension~$6$ and proved an upper bound of $23$ by subdividing the gallery into
cells and bounding the number of subsets that can be seen from within one cell. In the same paper he showed an upper bound for the VC-dimension of a 
gallery with holes of $O(\log h)$ where $h$ is the number of holes.

Since then there has been no progress on these general bounds. However, some variations of the original problem have been considered.
Isler et al. \cite{ikdv-vcdev-04} examined the case of exterior visibility. In this setting the points of~$S$ lie on the boundary of a polygon~$P$
and the ranges are sets of the form $vis(v)$ where $v$ is a point outside the convex hull of $P$. They showed that the VC-dimension is $5$.
The result can also be seen as a statement about wedges, as we will see later.
They also considered a more restricted version of exterior visibility where the view points $v$ all must lie on a circle around $P$, 
with VC-dimension $2$. For a $3$-dimensional version of exterior visibility with $S$ on the boundary of a polyhedron $Q$ they found that the
VC-dimension is in $O(\log n)$ as $n$ is the number of vertices of $Q$.
King \cite{k-vcdvt-08} examined the VC-dimension of visibility regions on polygonal terrains. 
For 1.5-dimensional terrains he proved that the VC-dimension equals $4$ and on 2.5-dimensional terrains there is no constant bound.
In \cite{gk-nrvsp-09} we considered the original setting and showed upper bounds of 13 for the number of points on the boundary and 15 for the number of points in convex position that can be shattered by interior visibility regions.

Without using the $\varepsilon$-net theorem, Kirkpatrick~\cite{k-ggnn-00} obtained a $64\cdot r \log \log r$ upper
bound to the number of {\em boundary} guards needed to cover the {\em boundary} of $P$. This raises the question 
if the factor $\log r$ in the $O(d\cdot r \log r)$ bound for $\epsilon$-nets in other geometric range spaces 
can be lowered to $\log \log r$ as well, as was shown to be true by Aronov et al.~\cite{bes-ssena-10} for 
special cases; see also King and Kirkpatrick~\cite{kk-iagsgp-10}.

\bigskip

\section{Proof Technique}         \label{technique-sec}%

Theorem~\ref{statement-theo} will be proven by contradiction. 
Throughout Sections~\ref{technique-sec} and~\ref{inner-sec}, we shall assume that there exists a simple
polygon $P$ containing a set $S$ of 15~points each of whose subsets is discernible.
That is, for each $T \subseteq S$ there is a view point $v_T$ in $P$ such that 
\begin{eqnarray}                            \label{prop1}
     T =   \mbox{vis}(v_T) \cap S 
\end{eqnarray}
holds, where, as usual,
   $  \mbox{vis}(v) = \{ x \in P;  \ \overline{xv} \subset P\} $
denotes the visibility domain of a point $v$ in the (closed) set $P$. 

We may assume that the points in $S$ and the view points $v_T$ are in general position,
by the following argument. If $p \notin T$,
then segment $\overline{v_Tp}$ is properly crossed by the boundary of $P$, that is, the segment and the complement
of $P$ have an interior point in common. On the other hand,
a visibility segment $\overline{v_Uq}$, where $q \in U$, can be touched by the boundary of $P$, because
this does just not block visibility. By finitely many, arbitrarily small local enlargements of $P$ we can remove
these touching boundary parts from the visibility segments without losing any proper crossing of a non-visibility segment.
Afterwards, all points and view points can be perturbed in small disks.

Property~\ref{prop1} can be rewritten as
\begin{eqnarray}                            \label{prop2}
        T=\{ p \in S; \ v_T \in \mbox{vis}(p)\}
\end{eqnarray}
This means, if we form the arrangement $Z$ of all visibility regions $\mbox{vis}(p)$, where $p \in S$,
then for each $T \subseteq S$ there is a cell (containing the view point $v_T$) which is contained
in exactly the visibility regions of the points in $T$.
To obtain a contradiction, one would like to argue that the number of cells in 
arrangement $Z$ is smaller than $2^{15}$, the number of subsets of $S$. But as we do not have
an upper bound on the number of vertices of $P$,  the complexity of $Z$ cannot be bounded from above.

For this reason we shall replace complex visibility regions with simple {\em wedges};
for wedge arrangements, a good upper complexity bound exists; see Theorem~\ref{isler-theo} below.
To illustrate this technique, let $a$ be a point of $S$. We assume that there are 
\begin{enumerate}
\item  points $b_1, b_2$ of $S$,
\item  a view point $v$ that sees $b_1$ and $b_2$, but not $a$, such that
\item  $a$ is contained in the  triangle defined by $\{v,b_1,b_2\}$;
\end{enumerate}
see Figure~\ref{wedge-fig}~(i). We denote by $U$ the wedge containing $v$ formed 
by the lines through $a$ and $b_1$ and $b_2$, respectively.
\begin{figure}[hbtp]%
  \begin{center}%
    \includegraphics[width=\textwidth,keepaspectratio]{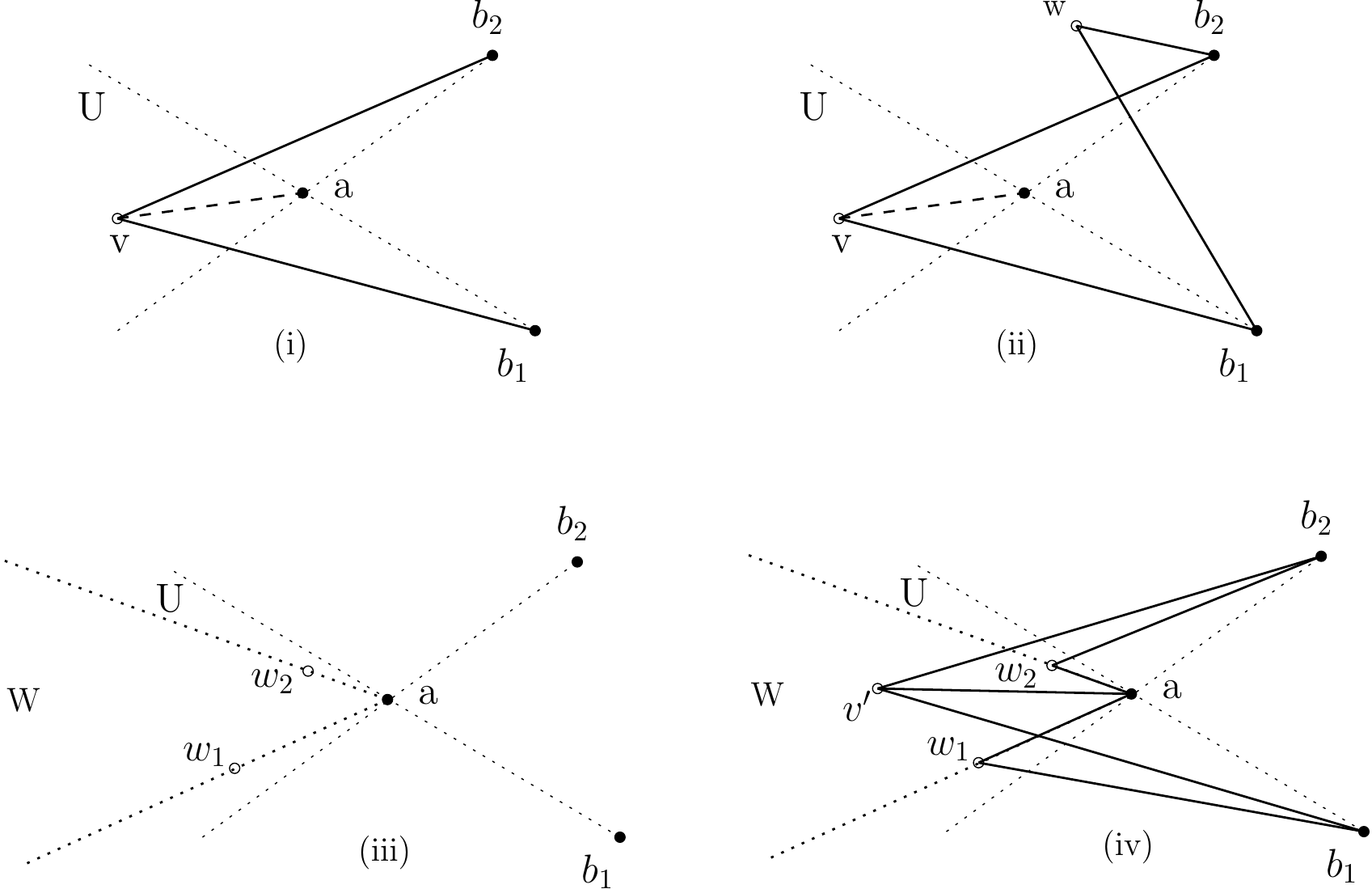}%
    \caption{Solid lines connect points that are mutually visible; such ``visibility segments''  
    are contained in polygon $P$. Dashed style indicates that the line of vision
    is blocked; these segments are crossed by the boundary of $P$.}%
    \label{wedge-fig}
  \end{center}%
\end{figure}
Any view point $w$ that sees $b_1$ and $b_2$ must be contained in wedge $U$. Otherwise, the chain of
visibility segments $v-b_1-w-b_2-v$ would encircle the line segment $\overline{va}$ connecting $v$ and $a$, 
preventing the boundary of $P$ from blocking the view from $v$ to $a$; see Figure~\ref{wedge-fig}~(ii). 

Let $w_1, w_2$ denote the outermost view points in $U$ that see $a, b_1, b_2$ and include a maximum angle
(by assumption, such view points exist; by the previous reasoning, they lie in $U$).
Then $w_1, w_2$ define a sub-wedge $W$ of $U$, as shown in Figure~\ref{wedge-fig}~(iii). 
We claim that in this situation
\begin{eqnarray}        \label{wedge-fact}
V_{\{b_1,b_2\}} \cap \mbox{vis}(a) = V_{\{b_1,b_2\}} \cap W
\end{eqnarray}
holds, where $V_{\{b_1,b_2\}}$ denotes the set of all view points that see at least $b_1$ and $b_2$.
Indeed, each view point that sees $b_1, b_2$ lies in $U$. If it sees $a$, too, it must lie in $W$,
by definition of $W$. Conversely, let $v'$ be a view point in $W$ that sees $b_1, b_2$. Then line segment
$\overline{v'a}$ is encircled by the visibility segments $v'-b_1-w_1-a-w_2-b_2-v'$,
as depicted in Figure~\ref{wedge-fig}~(iv). Thus, $v' \in \mbox{vis}(a)$.

Fact~\ref{wedge-fact} can be interpreted in the following way. We ``sacrifice'' two of the 15 points of $S$,
namely $b_1$ and $b_2$, and restrict ourselves to studying only those $2^{13}$ view points $V_{\{b_1,b_2\}}$ 
that see both $b_1, b_2$.  As a benefit, the visibility region $\mbox{vis}(a)$ behaves like a wedge
when restricted to $V_{\{b_1,b_2\}}$.

This technique will be applied as follows.
In Section~\ref{inner-sec} we prove, as a direct consequence, that at most~5 points can be situated {\em inside} the 
convex hull of $S$. Then, in Section~\ref{outer-sec}, we show that at most 9~points can be located {\em on} the
convex hull. 
Together, these claims imply Theorem~\ref{statement-theo}.

\bigskip

\section{Interior points}         \label{inner-sec}%

The goal of this section is in proving the following fact.

\begin{lemma}            \label{inner-lem}
At most five points of $S$ can lie inside the convex hull of $S$.
\end{lemma}
\begin{proof}
Suppose there are at least six interior points $a_i$ in the convex hull, $1 \leq i \leq k$.
Each of the remaining points of $S$ is a vertex of the convex hull of $S$.
Let $b_0, \ldots b_{m-1}$ an enumeration of these points in cyclic order. Let $v_B$ (where $B=\{b_0, \ldots, b_{m-1}\}$)
be the view point that sees only these vertices but no interior point; see Figure~\ref{inner-fig}.
\begin{figure}[hbtp]%
  \begin{center}%
    \includegraphics[scale=0.9,keepaspectratio]{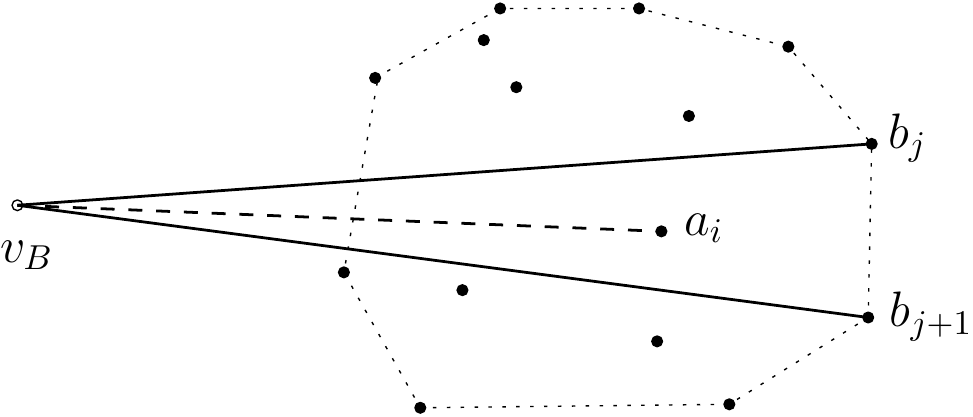}%
    \caption{Each interior point $a_i$ is contained in some triangle defined by \protect{$\{v_B,b_j,b_{j+1}\}$}.}%
    \label{inner-fig}
  \end{center}%
\end{figure}
Each interior point $a_i$ is contained in a triangle defined by $\{v_B,b_j,b_{j+1}\}$, for some~$j$ 
(where the indices are taken modulo $m$). 
Since properties~{1.--3.} mentioned in Section~\ref{technique-sec} are fulfilled, Fact~\ref{wedge-fact}
implies that there exists a wedge $W_i$ such that 
$V_{\{b_j,b_{j+1}\}} \cap \mbox{vis}(a_i) = V_{\{b_j,b_{j+1}\}} \cap W_i$ holds. 
If $V_B$ denotes the set of view points that see at least the points of $B$, 
 we obtain
\[
     V_B \cap \mbox{vis}(a_i) = V_B \cap W_i  \/ \mbox{ for } i=1, \ldots, 6,
\]
 which implies the following. For each subset $T$ of $A=\{a_1, \ldots, a_6\}$ the view point 
 $v_{T\cup B}$ lies in exactly those wedges $W_i$ where $a_i \in T$. But the arrangement
 of six or more wedges does not contain that many combinatorially different cells, as 
 an argument by Isler et al.~\cite{ikdv-vcdev-04} shows; see Theorem~\ref{isler-theo}. 
Thus, the convex hull of $S$ cannot contain six interior points.
\end{proof}

Therefore, at least~10 points of $S$ must be vertices of the convex hull of~$S$.

\begin{theorem}         \label{isler-theo}
(Isler et al.) For any arrangement of six or more wedges, there is a subset $T$ of wedges for which
no cell exists that is contained in exactly the wedges of $T$.
\end{theorem}
For convenience, we include a short proof based on the ideas in~\cite{ikdv-vcdev-04}.
\begin{proof}
By Euler's formula, an arrangement of $n$ wedges has $n+k+1$ many cells, where $k$
denotes the number of half-line intersections. Since two wedges intersect in at most 4~points---
in which case they are said to {\em cross} each other---we have
$k \leq 4 {n \choose {2}}$. Thus, an arrangement of $6$ wedges has at most $67$ cells.
We are going to provide an accounting argument which shows that for each wedge
one cell is missing from a maximum size arrangement (due to a shortage of intersections), or
one of the existing cells is redundant (because it stabs the same subset of wedges as some
other cell does). This will imply that at most $67-6 =61$ many of all $2^6=64$ different
subsets can be stabbed by a cell, thus proving the theorem.

Let $W$ be a wedge that is crossed by all other wedges, as shown in Figure~\ref{isler-fig}~(i).
Since the two shaded cells at the apex of $W$ and at infinity are both stabbing the subset $\{W\}$,
we can write off one cell of the arrangement as redundant, and exclude $W$ from further consideration.

The remaining $m$ wedges are used as the vertices of a graph $G$. Two vertices are connected
by an edge if their wedges do not cross. For each edge of $G$ there is one cell less in the arrangement,
as (at least) one of four possible intersection points is missing.
By construction, each vertex of $G$ has degree at least~1. Suppose that vertex $W$
is of degree~1, and let $W'$ denote the adjacent vertex in $G$.
If $W$ and $W'$ have at most two of four possible intersections,
even two cells are missing from the arrangement. If $W$ and $W'$ intersect in three points,
there is a redundant cell in $W$, in addition to the missing one; see Figure~\ref{isler-fig} (ii).
In either case, we may double the edge connecting $W$ and $W'$, as we obtain two savings
from this pair. In the resulting graph $H$ each vertex is of degree at least two. Thus,
$H$ contains at least $m$ edges, each of which represents a cell that is missing or redundant.
\begin{figure}[hbtp]%
 \begin{center}%
   \includegraphics[scale=0.9,keepaspectratio]{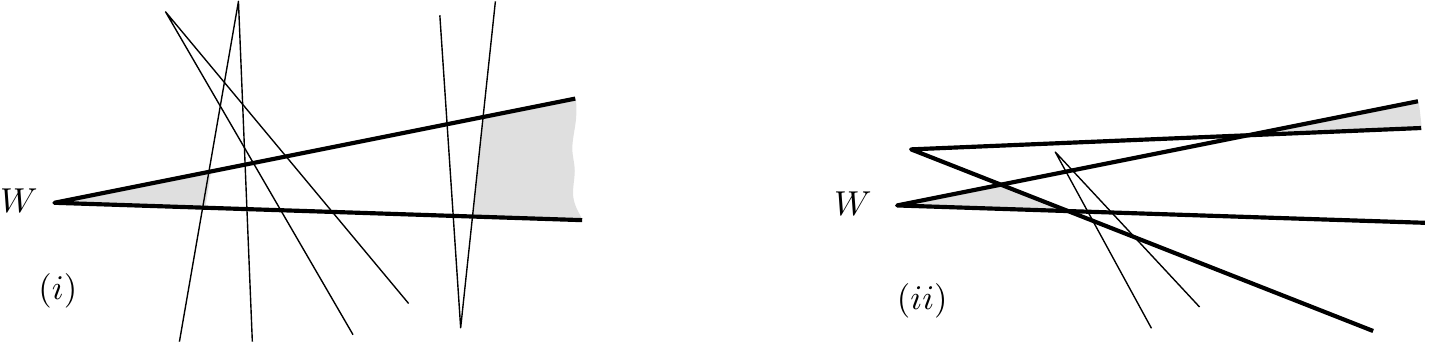}%
   \caption{In (i) and (ii), respectively, the shaded cells are contained in wedge $W$ only.}%
   \label{isler-fig}
 \end{center}%
\end{figure}
\end{proof}

\bigskip
\section{Points on the boundary of the convex hull}         \label{outer-sec}%
Ignoring interior points, we prove, in this section, the following fact.

\begin{lemma}            \label{outer-lem}
Let $S$ be a set of~10 points in convex position inside a simple polygon, $P$. 
Then not all of the subsets of $S$ are discernible.
\end{lemma}
\begin{proof}
Again, the proof is by contradiction. So let $S$ be a set of~10 points
in convex position inside a simple polygon $P$. Assume that every subset of
$S$ is discernible.

First, we enumerate the points around the convex hull.%
\footnote{The edges of the convex hull of $S$ may intersect the boundary of $P$.
} 
Let $E$ denote the set of even indexed points. 
Let $v_E$ be the view point that sees exactly the even indexed points. If $v_E$ lies outside the convex hull, 
$\mbox{ch}(S)$, of $S$, we draw the 
two tangents from $v_E$ at $\mbox{ch}(S)$. The points between the two tangent points 
facing $v_E$ are called {\em front points}, all other
points are named  {\em back points} of $S$; see Figure~\ref{front-fig}.
(If $v_E \in \mbox{ch}(S)$ then all points of $S$ are called back points.)

\begin{figure}[htbp]%
  \begin{center}%
    \includegraphics[scale=0.7,keepaspectratio]{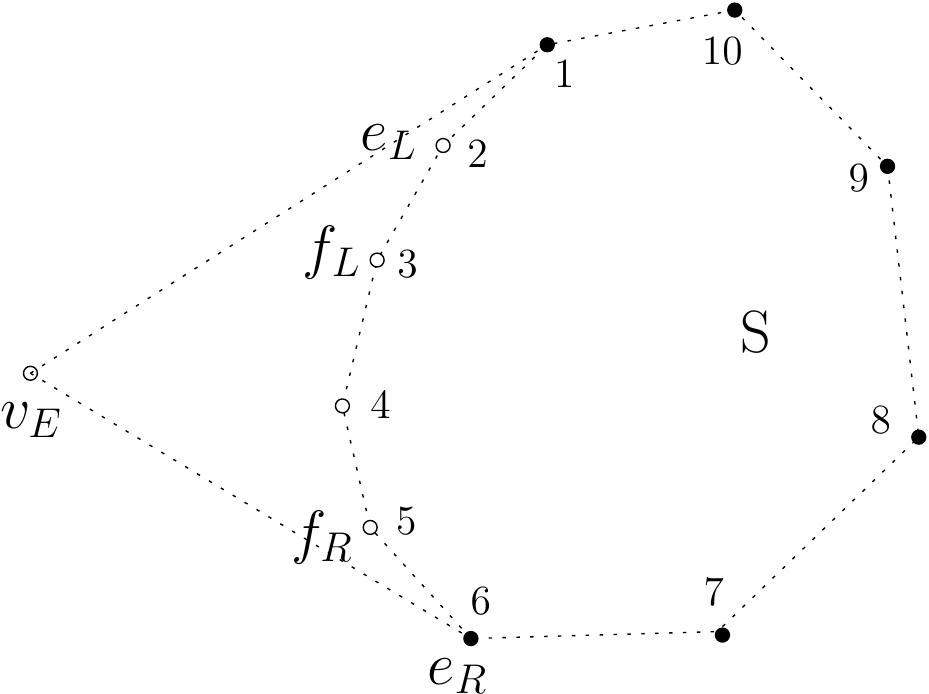}%
    \caption{Front points appear in white, back points in black. View point $v_E$
    sees exactly the points of even index.}%
    \label{front-fig}
  \end{center}%
\end{figure}

We are going to discuss the case depicted in Figure~\ref{front-fig} first, namely:

\noindent{\bf Case~1: There exists an odd front point.}

It follows from the definition of front points that in this case $v_E$ lies outside the convex hull of $S$.
Let $f_L$ and $f_R$ be the outermost left and right front points with odd index, as seen from $v_E$;
and let $e_L$ and $e_R$ denote their outer neighbors, as shown in Figure~\ref{front-fig}.
While $f_L=f_R$ is possible, we always have $e_L \not= e_R$. Observe that $e_L$ and
$e_R$ may be front or back points; this will require some case analysis later on.

{\bf Notation.}
For two points $a,b$, let $H^+(a,b)$  denote the open half-plane to the left of the
ray $L(a,b)$ from $a$ through $b$, and $H^-(a,b)$ the open half-plane to its right.
\begin{figure*}[hbtp]%
  \begin{center}%
    \includegraphics[width=\textwidth,keepaspectratio]{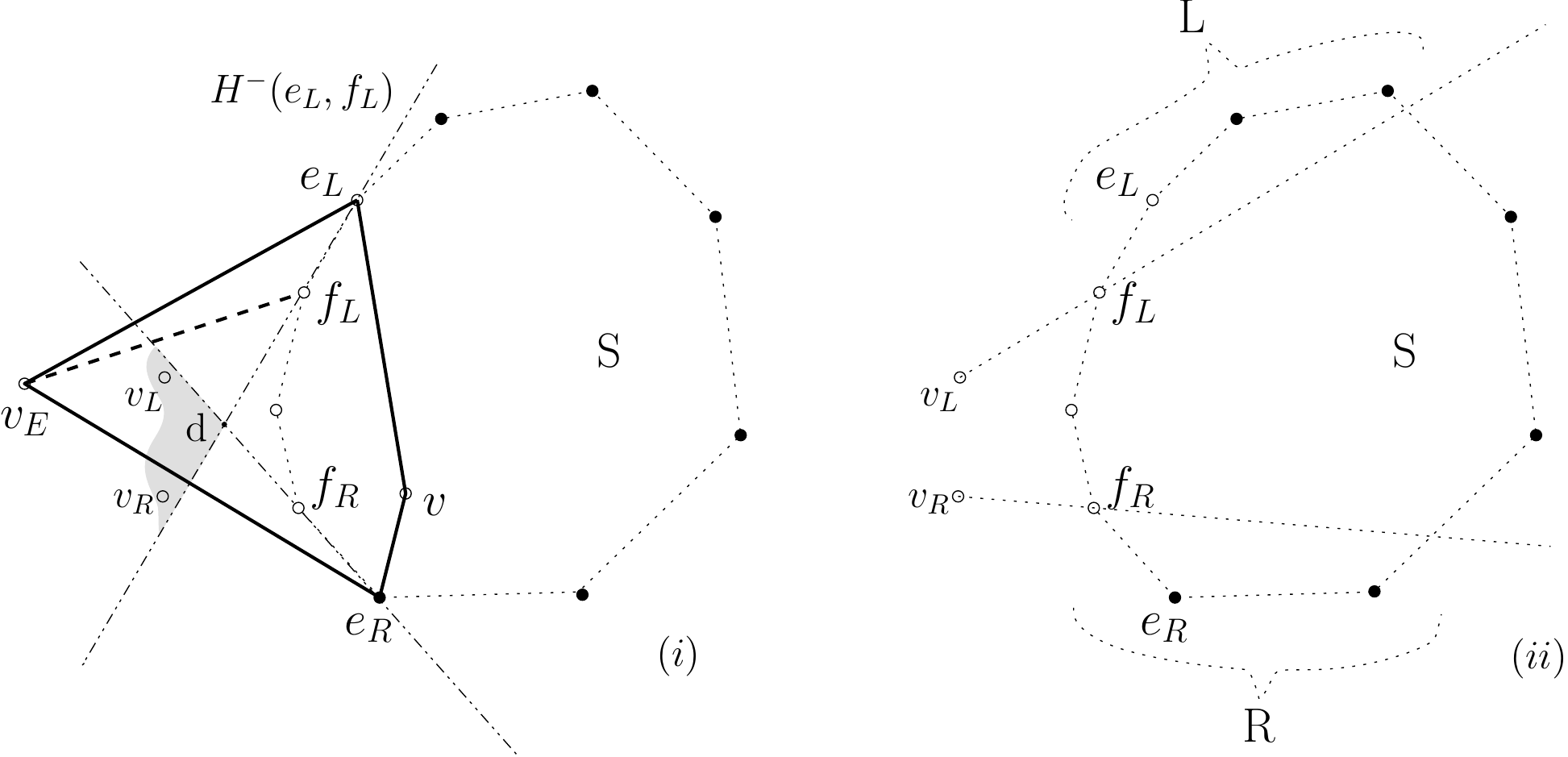}%
    \caption{(i) As segment $\overline{v_Ef_L}$ must be intersected by the boundary of $P$, it cannot
    be encircled by visibility segments. (ii) Defining subsets $L$ and $R$ of $S$.}%
    \label{front2-fig}
  \end{center}%
\end{figure*}

\begin{claim}         \label{newfront-claim}
Each view point $v$ that sees $e_L$ and $e_R$ lies in $H^-(e_L,f_L) \cap H^-(f_R,e_R)$.
\end{claim}
\begin{proof}
If $v$ were contained in $H^+(e_L,f_L)$ then the chain of visibility segments 
$e_L-v-e_R-v_E-e_L$ would encircle the segment $\overline{v_Ef_L}$---a contradiction,
because $v_E$ does not see the odd indexed point $f_L$; see  Figure~\ref{front2-fig}~(i).
\end{proof}

We now define two subsets $L$ and $R$ of $S$ that will be crucial in our proof.
\begin{definition}    \label{lr-defi}
(i) Let $v_L:=v_{S\setminus \{f_L\}}$ and $v_R:=v_{S\setminus \{f_R\}}$
denote the view points that see all of $S$ except $f_L$ or $f_R$, respectively.\\
(ii) Let $L:= S \cap H^+(v_L, f_L)$ and $R:= S \cap H^-(v_R, f_R)$.
\end{definition}
By Claim~\ref{newfront-claim}, the points of $S$ contained in the triangle $(e_R, e_L, v_E)$ are front points with respect
to $v_R, v_L$, too; see Figure~\ref{front2-fig}.
\begin{claim}         \label{lr-claim}
None of the sets $L, R, S\setminus (L \cup R)$ are empty. The sets $L$ and $ R$ are disjoint. 
\end{claim}
\begin{proof}
By construction, we have $e_L \in L$, $e_R \in R$, and $f_L, f_R \not\in L \cup R$.
If $v_L=v_R$ then $L \cap R = \emptyset$, obviously. Otherwise, there is at least one
even indexed point, $e$, between $f_L$ and $f_R$ on $\mbox{ch}(S)$. Assume 
that there exists a point $q$ of $S$ in the intersection of $L$ and $R$. Then segment
$\overline{v_Rf_R}$ would be encircled by the visibility chain $q-v_R-e-v_L-q$, 
contradicting the fact that $v_R$ sees every point {\em but} $f_R$; see Figure~\ref{lrdis-fig}. 
\end{proof}
\begin{figure}[hbtp]%
  \begin{center}%
    \includegraphics[scale=0.6,keepaspectratio]{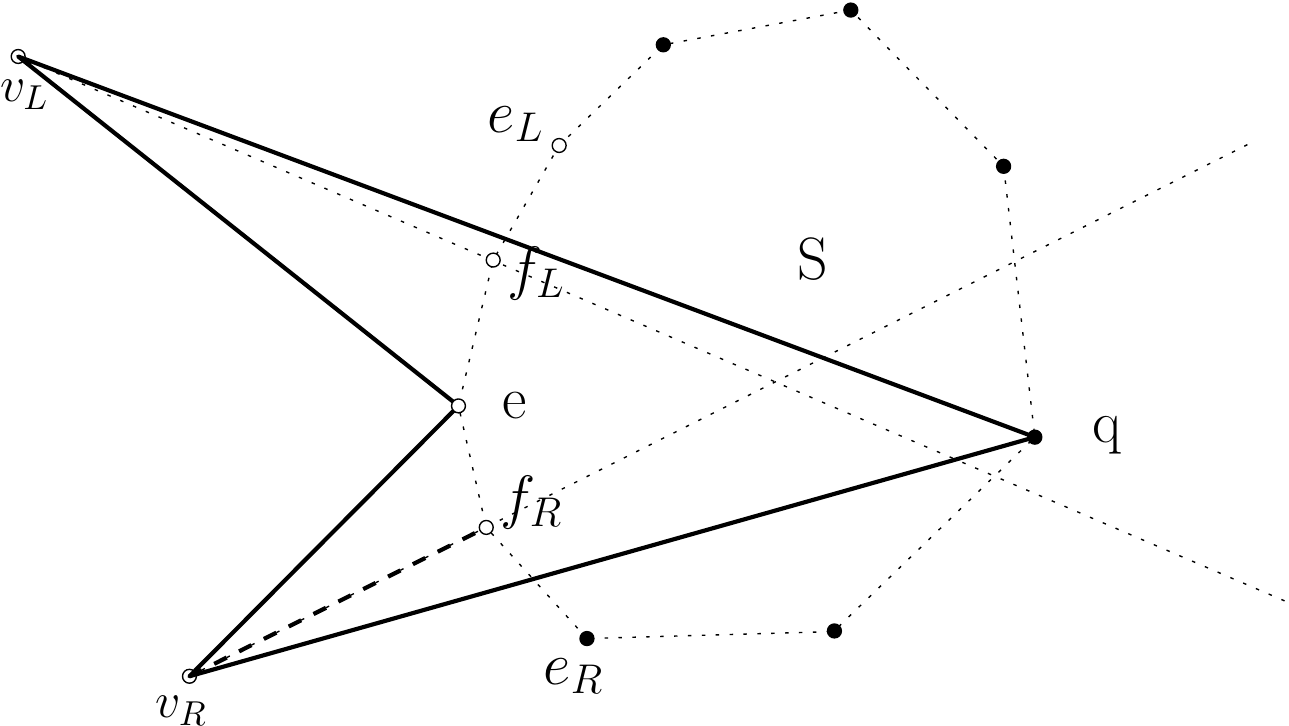}%
    \caption{$L$ and $R$ are disjoint.}%
    \label{lrdis-fig}
  \end{center}%
\end{figure}

The purpose of the sets $L$ and $R$ will now become clear: They contain points 
like $b_1, b_2$ in Section~\ref{technique-sec}, that help us reduce visibility regions
to wedges. The precise property will be stated for $R$ in Lemma~\ref{case-lem};
a symmetric property holds for $L$. The proof of Lemma~\ref{case-lem} will be 
postponed. First, we shall derive a conclusion in Lemma~\ref{main-lem},
and use it in completing the proof of Lemma~\ref{outer-lem} in Case~1.

\begin{lemma}      \label{case-lem}
There exist points $r_1, r_2$ in $R$ such that the following holds either for $Q=\mbox{vis}(r_1) \cap \mbox{vis}(r_2)$ 
or for $Q= \mbox{vis}(r_1)^c \cap \mbox{vis}(r_2) $.
For each $p \in S$ different from $r_1, r_2$, 
each view point that (i) sees $p$, (ii) lies in $Q$, and (iii) sees at least one point of $L$, 
is contained in the half-plane $H^-(p,r_2)$.
\end{lemma}

Here, $D^c$ denotes the complement of a set $D$.
A symmetric lemma holds for points $l_1, l_2 \in L$, a set 
$Q' \in \{  \mbox{vis}(l_1) \cap \mbox{vis}(l_2), \  \mbox{vis}(l_1)^c \cap \mbox{vis}(l_2) \}$
and the half-plane $H^-(l_2,p)$. Adding up these facts yields the following.

\begin{lemma}     \label{main-lem}
Let $p \in S \setminus \{l_1, l_2, r_1, r_2\}$. Then each view point in $Q \cap Q'$ that 
sees $p$ lies in the wedge $U_p = H^-(p,r_2) \cap H^-(l_2,p)$.
\end{lemma}

Now we can proceed as in Section~\ref{technique-sec}; see Figure~\ref{wedge-fig}~(iii) and (iv).
Within wedge $U_p$ we find a sub-wedge $W_p$ satisfying
\begin{eqnarray}      \label{conclu}
   Q \cap Q'  \cap \mbox{vis}(p) =  Q \cap Q'  \cap  W_p,
\end{eqnarray} 
with the same arguments that led to Fact~\ref{wedge-fact}, replacing $(a, b_1, b_2)$ with
$(p, r_2, l_2)$. Since membership in $Q, Q'$ only prescribes the visibility of $\{l_1, l_2, r_1, r_2\}$,
Fact~\ref{conclu} implies the following. For each subset 
$T \subseteq S \setminus \{l_1, l_2, r_1, r_2\}$ there exists a cell in the arrangement of the remaining 
six wedges $W_p$, where $p\in S\setminus\{l_1,l_2,r_1,r_2\}$, that is contained in precisely the wedges related to $T$. 
As in Section~\ref{inner-sec}, this contradicts Theorem~\ref{isler-theo} and proves Lemma~\ref{outer-lem} in Case~1.
\end{proof}
It remains to show how to find $r_1, r_2$ and $Q$ in Lemma~\ref{case-lem}. 

\begin{proof} (of Lemma~\ref{case-lem}) Before starting a case analysis depending on properties of
$R$ and $e_R$ we list some helpful facts.

\begin{claim}         \label{prep1-claim}
If a view point $v$ sees a point $r \in R$ and a point $s \notin R \cup \{f_R\}$ then
$v \in H^-(s,r)$. A symmetric claim holds for $L$.
\end{claim}
\begin{proof}
Otherwise, $\overline{v_Rf_R}$ would be encircled by $r-v-s-v_R-r$, 
since $f_R$ lies in the triangle defined by $(v_R,r,s)$; see Figure~\ref{prep-fig}~(i).
\end{proof}
\begin{figure*}[hbtp]%
  \begin{center}%
    \includegraphics[width=\textwidth,keepaspectratio]{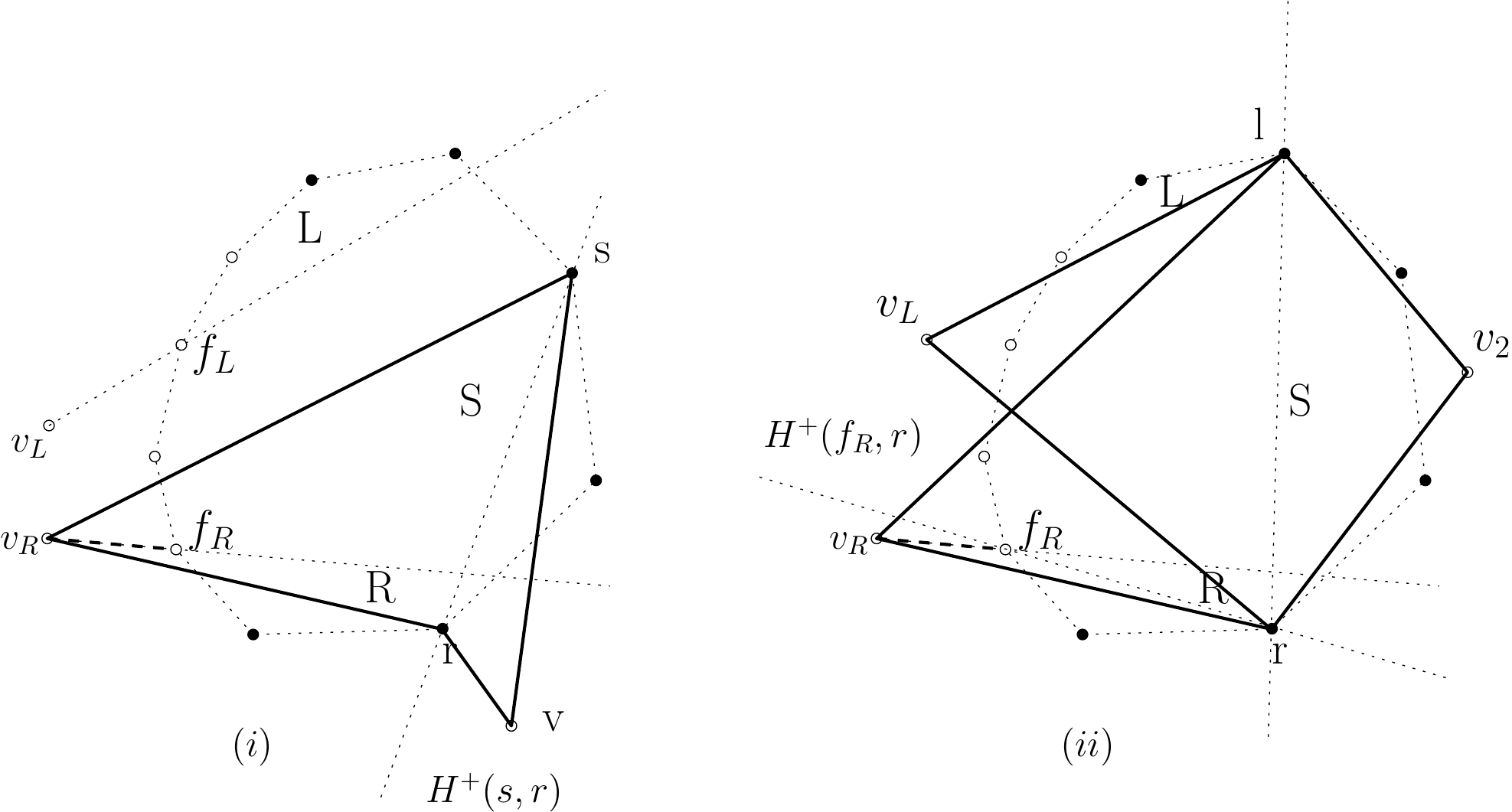}%
    \caption{Illustration to Claims~\ref{prep1-claim} and~\ref{prep2-claim}.}%
    \label{prep-fig}
  \end{center}%
\end{figure*}
The next fact narrows the locus from which two points, one from $L$ and $R$ each,
are visible. 
\begin{claim}         \label{prep2-claim}
If a view point $v$ sees points $r \in R$ and $l \in L$ then
$v$ lies in the wedge $H^-(f_R,r) \cap H^-(l,f_L)$, and on the same side of $L(r,l)$
as $v_R$ and $v_L$ do.
\end{claim}
\begin{proof}
If $v \in H^+(f_R,r)$, or if $v$ were situated on the opposite side of $L(r,l)$, 
then $\overline{v_Rf_R}$ would be encircled by $r-v-l-v_R-r$; see points $v=v_1$ and $v=v_2$ 
in Figure~\ref{prep-fig}~(ii).
\end{proof}

Now we start on the case analysis. In each case, we need to define $r_1, r_2 \in R$ and 
a set $Q=\mbox{vis}(r_1) \cap \mbox{vis}(r_2)$ or $Q=\mbox{vis}(r_1)^c \cap \mbox{vis}(r_2)$.
Then we must prove that the following assertion of Lemma~\ref{case-lem} holds.

{\em {\bf Assertion}\\
 If $p \in S$ is different from $r_1, r_2$, and if $v \in Q$ is a view point that sees $p$ and
some  point $l \in L$, then $v \in H^-(p,r_2)$.}

{\bf Case 1a:} Point set $R$ contains at most two points.\\
We define $\{r_1, r_2\}:=R$ and let $Q:=  \mbox{vis}(r_1) \cap \mbox{vis}(r_2)$.\\
Let $p$ and $v$ be as in the Assertion.
If $p \not= f_R$ then Claim~\ref{prep1-claim} implies $v \in H^-(p,r_2)$. 
If $p= f_R$ we obtain $v \in H^-(p,r_2)$  by the first statement in Claim~\ref{prep2-claim}.
\begin{figure}[hbtp]%
  \begin{center}%
   \includegraphics[width=\textwidth,keepaspectratio]{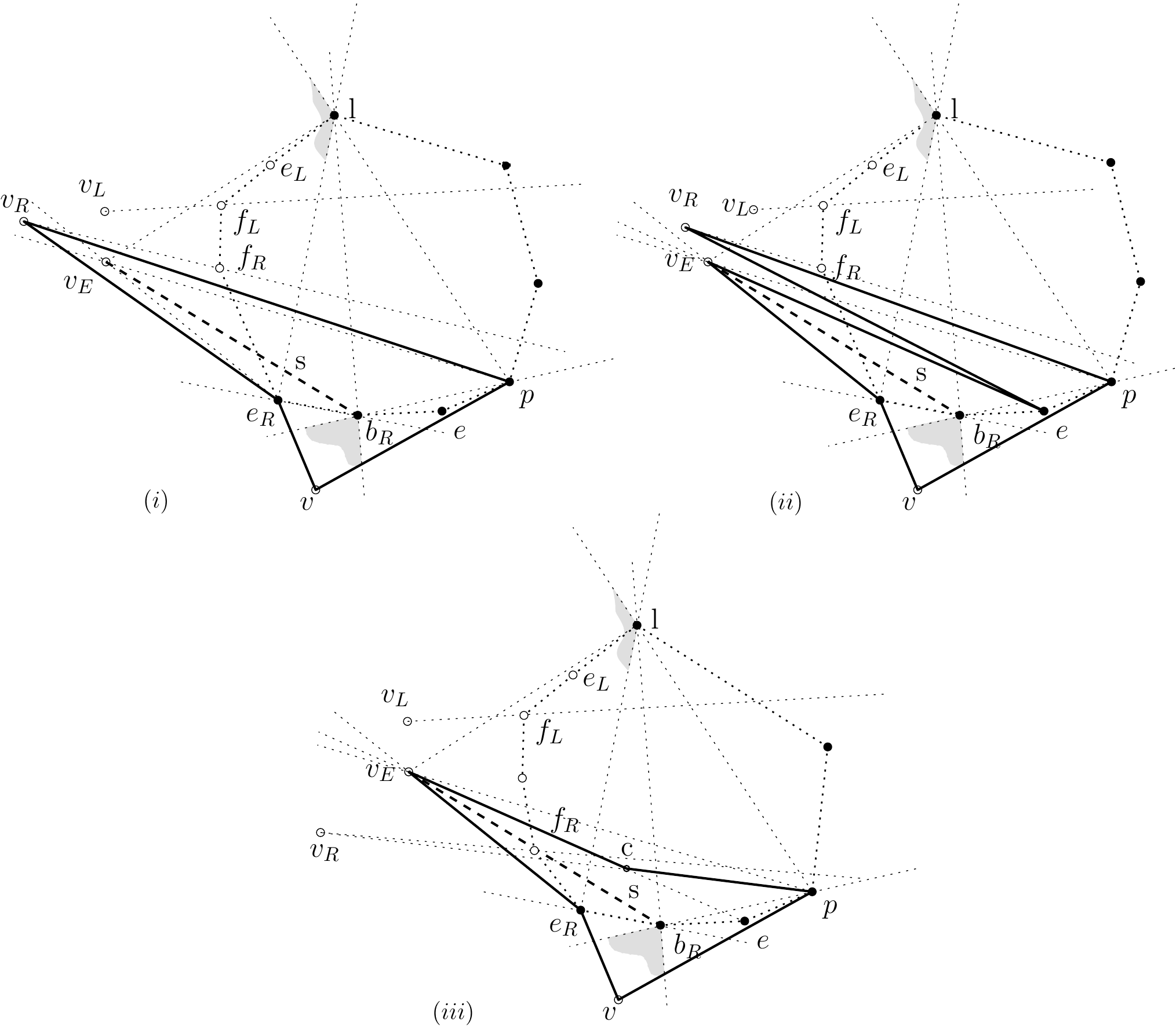}%
    \caption{Illustrations of Case 1b.}%
    \label{case1b-fig}
  \end{center}%
\end{figure}

{\bf Case 1b:} Point set $R$ contains more than two points, and $e_R$ is tangent point of $\mbox{ch}(S)$
as seen from $v_E$; compare Figure~\ref{front-fig}.\\
We set $r_1:=e_R$ and let $r_2$ be the odd indexed back point $b_R$ counterclockwise next to $e_R$.
Moreover, $Q:=  \mbox{vis}(r_1) \cap \mbox{vis}(r_2)$.\\
For each $p \notin R$ the proof of Case~1a applies. Let $p \in R$ be different from $r_1, r_2$.
Assume, by way of contradiction, that $v\in H^+(p,r_2)$ holds. Since the second statement of 
Claim~\ref{prep2-claim} implies
$v \in H^-(l,e_R) \cap H^-(l,p) \subset H^-(l,b_R)$, we obtain $v \in H^-(l,b_R) \cap H^+(p,b_R)$;
see Figure~\ref{case1b-fig}.
Now we discuss the location of view point $v_R$. If it lies in the wedge $H^+(e_R, v_E) \cap H^+(v_E,p)$ 
then segment $s:=\overline{v_Eb_R}$ is encircled by $e_R-v_R-p-v-e_R$; see Figure~\ref{case1b-fig}~(i).
If $v_R$ does not lie in this wedge, let $e$ be the counterclockwise neighbor of $b_R$ in $R$.
If $v_R$ lies on the same side of $L(e,v_E)$ as $p$, then 
$e_R-v_E-e-v_R-p-v-e_R$ protects segment $s$; see~(ii).
If it lies on the opposite side, then $\overline{v_Ee}$ intersects  $\overline{v_Rp}$
at some point $c$, and $e_R-v_E-c-p-v-e_R$ encircles segment~$s$; see (iii).
In either situation, we obtain a contradiction.

Before continuing the case analysis we prove a simple fact.
\begin{lemma}          \label{tria-lem}
Let $a,b,c$ denote the vertices of a triangle, in counterclockwise order. Suppose there exists 
a view point $w$ in $H^+(b,a) \cap H^-(c,b)$ that sees $a$ and $c$. Then,
each view point $v \in H^+(b,a)$ that sees $a$ and $c$ but not $b$ lies in $H^-(c,b)$.
\end{lemma}
\begin{proof}
Otherwise, segment $\overline{vb}$ would be encircled by $c-v-a-w-c$; see Figure~\ref{tria-fig}.
\end{proof}
\begin{figure}[hbtp]%
  \begin{center}%
   \includegraphics[scale=0.7,keepaspectratio]{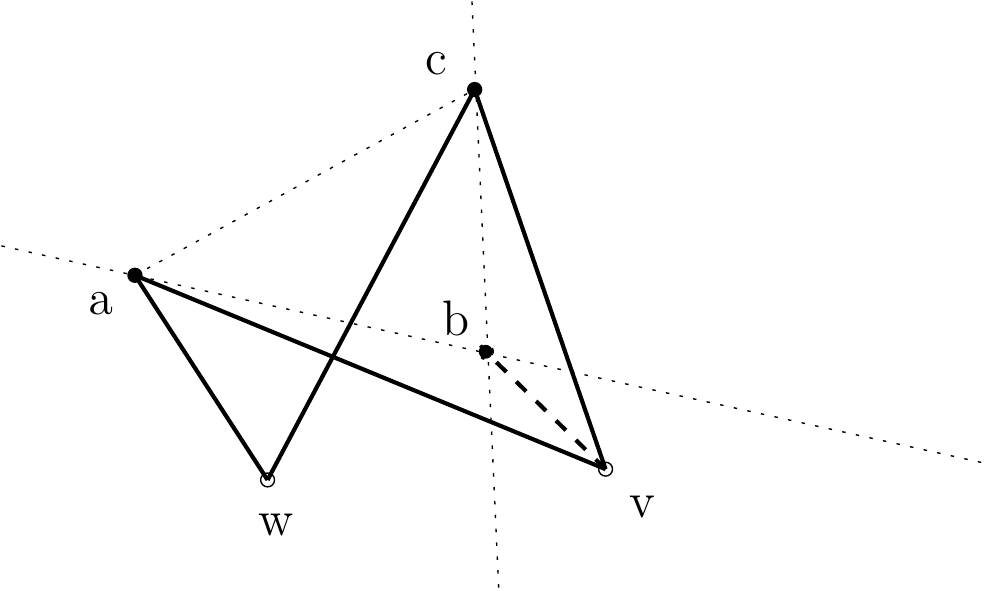}%
    \caption{Proof of Lemma~\ref{tria-lem}}%
    \label{tria-fig}
  \end{center}%
\end{figure}
{\bf Case 1c:} Point set $R$ contains more than two points, and the 
counterclockwise neighbor, $b_R$, of $e_R$, is tangent point as seen from $v_E$.
Let $e$ denote the counterclockwise neighbor of $b_R$, and let $w_R := v_{S \setminus \{e_R\}}$
denote the view point that sees all of $S$ except $e_R$.
We consider three subcases, depending on the location of $w_R$.

(1ci) If $w_R \in H^-(b_R, e_R)$, we set 
$$(r_1,r_2,Q):=(e_R,b_R,\mbox{vis}(e_R)^c \cap \mbox{vis}(b_R)).$$
To prove the Assertion, let $p \not= e_R, b_R$, and let $v$ be a view point that 
sees $p, b_R, l$ but not $e_R$, for some $l \in L$. 

As both $w_R$ and $v$ see $l \in L$ and $b_R \in R$, 
Claim~\ref{prep2-claim} implies $w_R, v \in H^-(f_R,b_R) \cap H^-(l,f_L) \subset H^+(e_R,l)$.
The latter inclusion allows us to apply Lemma~\ref{tria-lem} to $(a,b,c,w)=(l,e_R,b_R,w_R)$,
which yields $v \in H^-(b_R,e_R)$. Now $v \in H^-(p,b_R)$ follows; 
see Figure~\ref{case1c-fig}~(i).
\begin{figure*}[hbtp]%
  \begin{center}%
   \includegraphics[width=\textwidth,keepaspectratio]{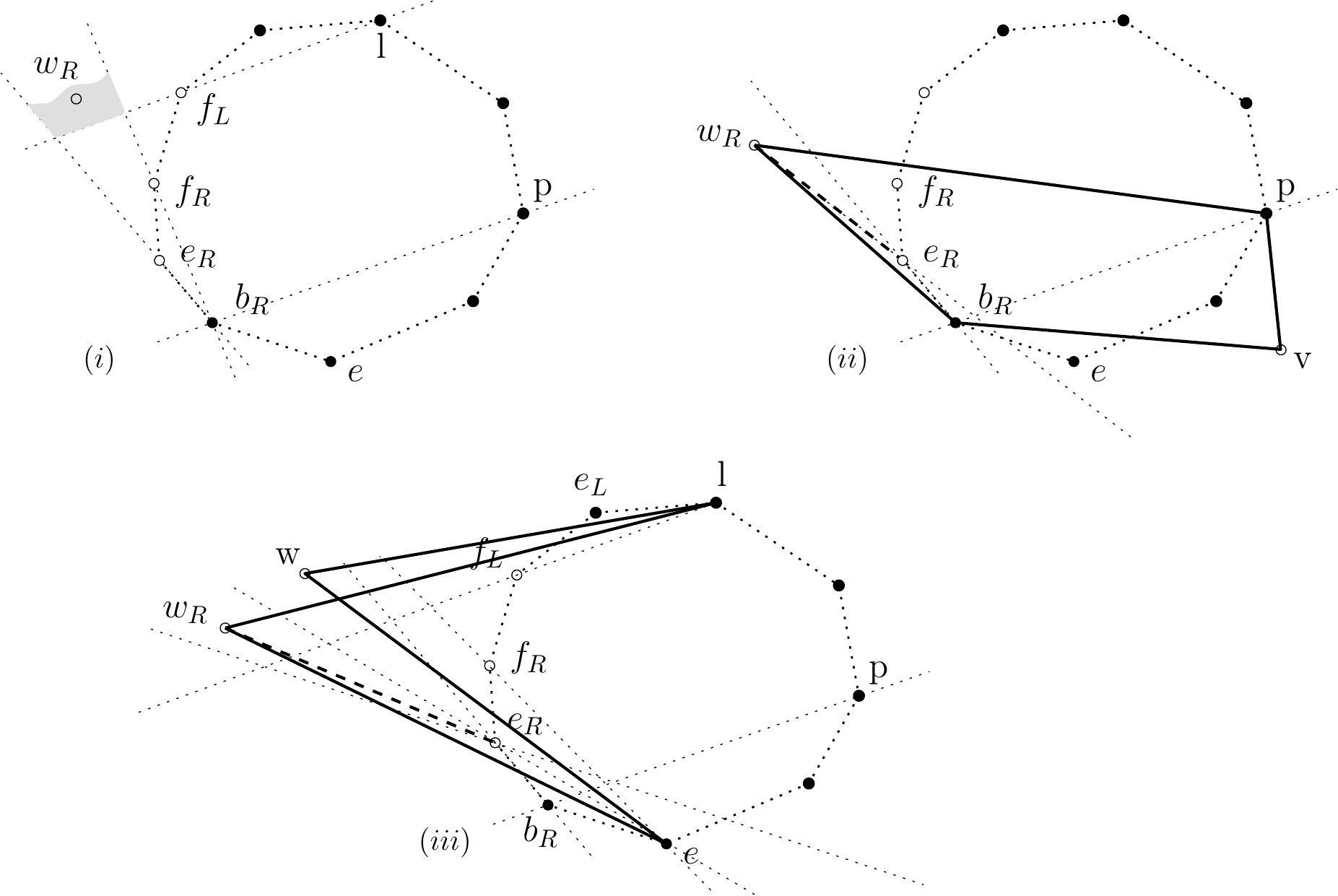}%
    \caption{Illustrations of Case 1c.}%
    \label{case1c-fig}
  \end{center}%
\end{figure*}

(1cii)  If $w_R \in H^+(b_R, e_R)$, and if $b_R$ and $e$ are situated on opposite sides
of $L(w_R,e_R)$, we set  $(r_1,r_2,Q):=(e_R,b_R,\mbox{vis}(e_R) \cap \mbox{vis}(b_R))$.\\
All points of $S':=S\setminus \{e_R,b_R\}$ lie on the same side of $L(w_R,e_R)$ as $e$.
A view point $v$ that sees some point $p \in S'$ and $b_R$ must be in $H^-(p,b_R)$.
Otherwise, $\overline{w_Re_R}$ would be encircled by $b_R-w_R-p-v-b_R$; 
see Figure~\ref{case1c-fig}~(ii).

(1ciii)  If $w_R \in H^+(b_R, e_R)$, and if $b_R$ and $e$ are situated on the same side
of $L(w_R,e_R)$, we set  $(r_1,r_2,Q):=(b_R,e,\mbox{vis}(b_R)^c \cap \mbox{vis}(e))$.\\
Clearly, $w_R \in H^+(e,e_R)$. Each view point $w$ that sees $e$ and some $l \in L$
---in particular point $v$ of the Assertion--- must 
lie in $H^-(e_R, e)$, or $\overline{w_Re_R}$ would be enclosed by $w_R-l-w-e-w_R$;
see Figure~\ref{case1c-fig}~(iii) (observe that $w$ must be contained in 
$H^-(f_R, e) \cap H^-(l,f_L) \subset H^+(e,l)$, 
by Claim~\ref{prep2-claim}, as depicted in the figure).

Let $x$ denote the view point that sees exactly $e_R, e, e_L, l$. By Claim~\ref{prep2-claim},
$x \in H^+(e_R,e_L) \cap H^+(e, e_L) \subset H^+(b_R,e_L)$.
We file for later use that $x \in H^+(b_R,l)$ holds, for the same reason.
Since $b_R$ is tangent point from $v_E$, we have $v_E \in H^-(e,b_R)$. Thus, we
can apply Lemma~\ref{tria-lem} to $(a,b,c,w)=(e_L, b_R, e, v_E)$ and obtain
$x \in H^-(e,b_R)$. 

We have just shown that $x \in H^+(b_R,l) \cap H^-(e,b_R)$ holds. Moreover,
Claim~\ref{prep2-claim} implies $v \in H^-(f_R,e) \cap H^-(l,f_L) \subset H^+(b_R,l)$
since $v$ sees $l$ and $e$. Since $v$ does not see $b_R$ 
we can apply Lemma~\ref{tria-lem} to $(a,b,c,w)=( l,b_R,e,x)$ and obtain $v \in H^-(e, b_R)$.
Together with the first finding in (1ciii), this implies
$v \in H^-(e_R, e) \cap H^-(e, b_R) \subset H^-(p,e)$ for all $p \not= b_R, e$.

This completes the proof of Lemma~\ref{outer-lem} in Case~1.

\end{proof}
Now we discuss the second case of Lemma \ref{outer-lem}, thereby completing its proof.
This also completes the proof of our main result, Theorem~\ref{statement-theo}.

\bigskip

{\bf Case~2: There is no odd front point.}

In this situation, view point $v_E$ either lies inside $ch(S)$, so that no front point exists, or $v_E$ lies outside $ch(S)$,
and at most one front point is visible from $v_E$ between the two tangent points on $ch(S)$; if so, its index is even. 

Independently of the position of $v_E$,  we introduce some notation.
Let $v_S$ denote the point that sees all points in $S$. 
The line $G$ through $v_E$ and $v_S$ divides
$S$ into two subsets, $L$ and $R$ (not to be confused with $L$ and $R$ in case 1), one of which may possibly be empty.
We cut $G$ at $v_E$, and rotate the half-line passing through $v_S$ over $L$; see
Figure~\ref{Nomenklatur3-fig}. The first and the last odd indexed points of $L$ encountered
during this rotational sweep are named $l_1$ and $l_2$, respectively. 
Similarly, $r_1$ and $r_2$ are defined in $R$.

We observe that, e.~g., $l_1$ and $l_2$ need not exist, or that $l_1=l_2$ may hold; these cases will be taken
care of in the subsequent analysis. Also, the half-line rotating about $v_E$ may cut through $S$ in its
start position, depending on the position of $v_S$. This is of no concern for our proof, which is literally the
same for either situation.
\begin{figure}[hbtp]%
\begin{minipage}[t]{0.4\textwidth}
\includegraphics[width=\textwidth,keepaspectratio]{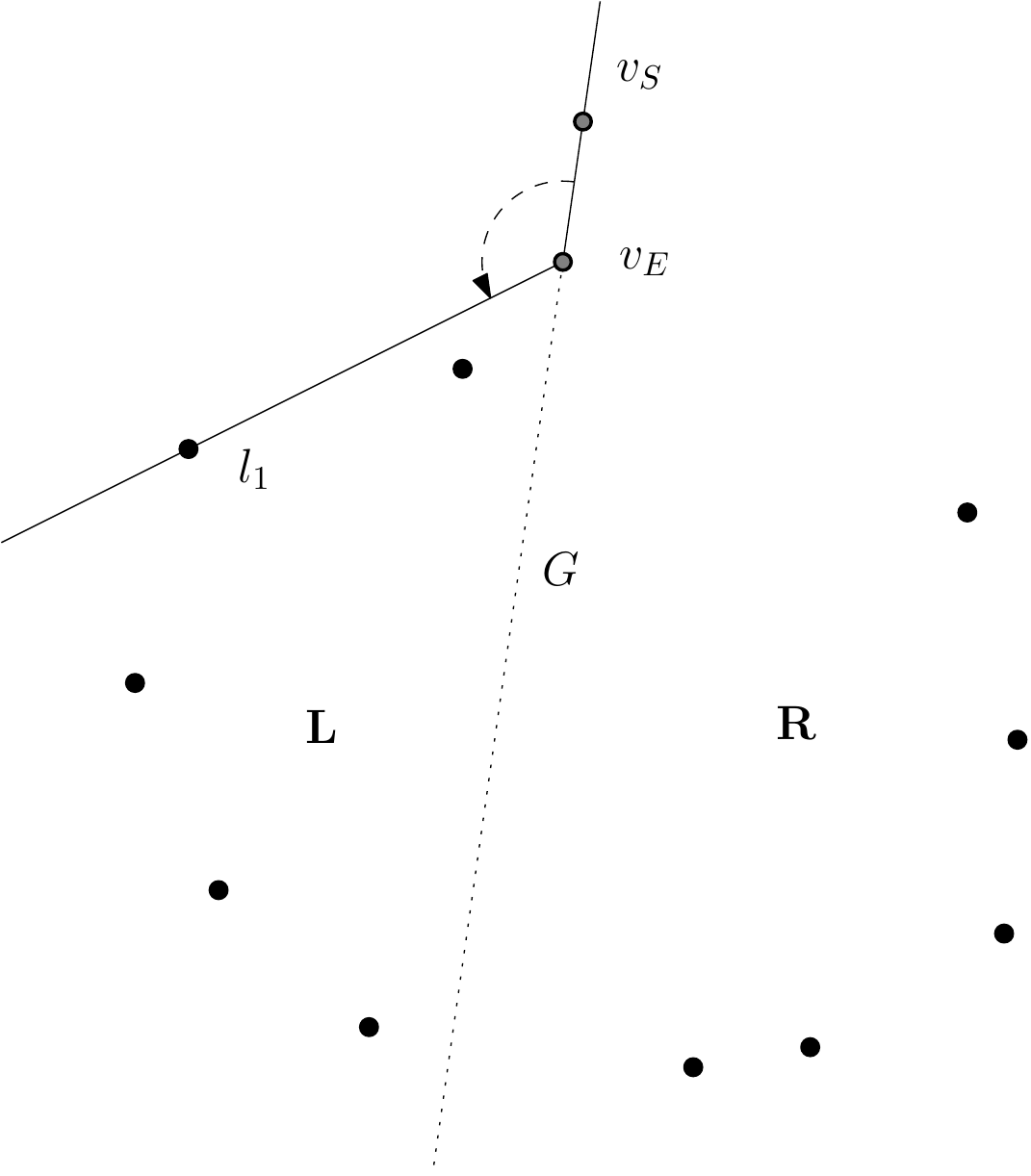}%
\end{minipage}
\hspace{2cm}
\begin{minipage}[t]{0.4\textwidth}
\includegraphics[width=\textwidth,keepaspectratio]{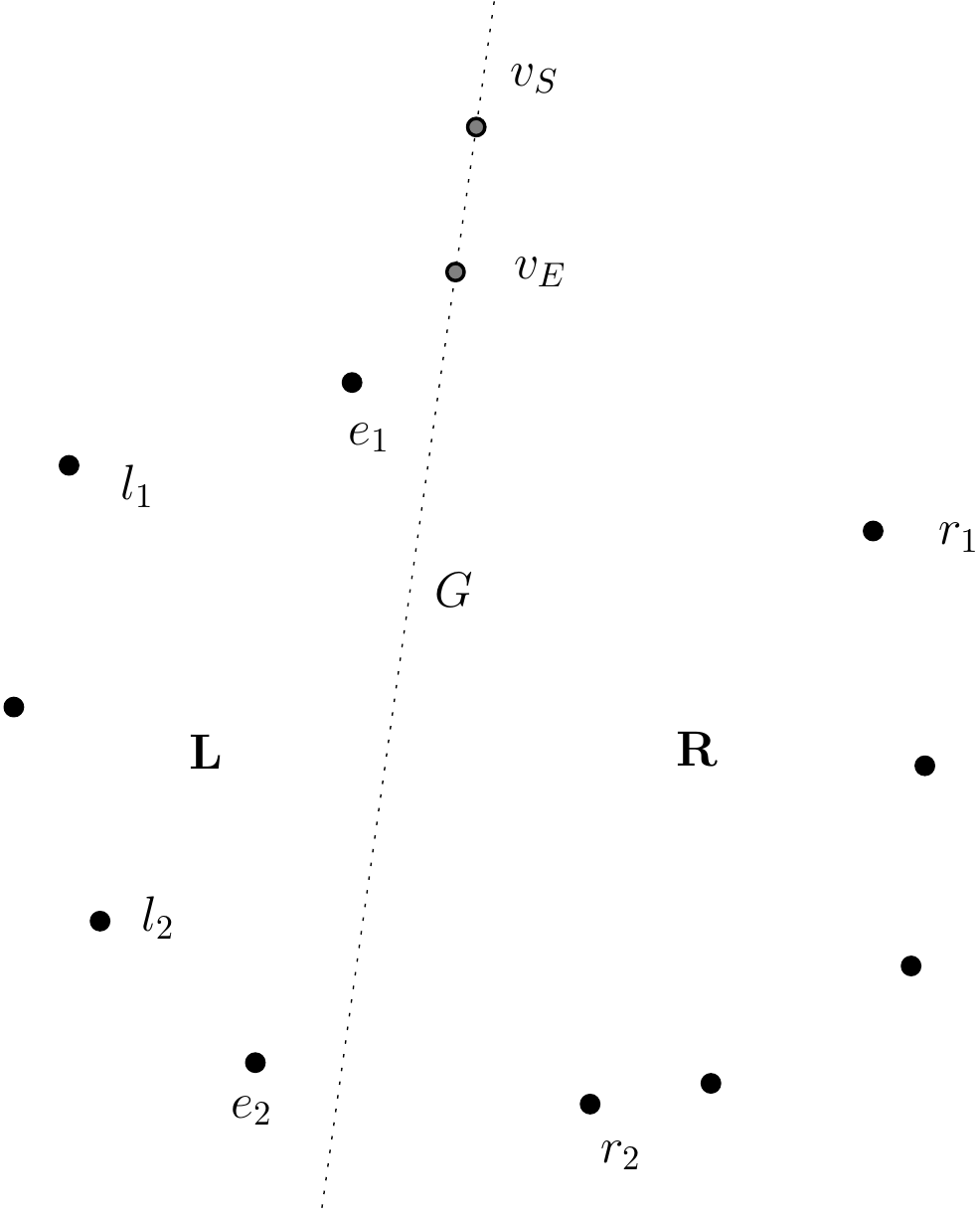}%
\end{minipage}
\caption{The half-line is rotated about $v_E$ over $L$. The first odd point encountered is named $l_1$, the last one $l_2$. }
\label{Nomenklatur3-fig}
\end{figure}
\begin{lemma} If there are odd-indexed points in both $L$ and $R$, then exactly one point lies 
between $l_1$ and $r_1$ on the boundary of the convex hull of $S$.
This point has even index and will be called $e_1$. Similarly, there is exactly one even-indexed point between $l_2$ and $r_2$, 
called $e_2$.
\end{lemma}

\begin{proof}
If there was more than one point on $\partial ch (S)$ between $l_1$ and $r_1$, 
one of them would have an odd index.
Let w.l.o.g $o$ lie on the same side of $L(v_E,v_S)$ as $l_1$. Being odd, $l_1$ and $o$ must be back points, since no
odd front points exist in Case~2.
As the order in which back points of $L$ are encountered by the rotating rays 
coincides with their order on the boundary of the convex hull,
$o$ would have to be hit by the rotating ray before $l_1$, contradicting the definition of $l_1$. 
The same argument applies to $l_2$ and $r_2$.
\end{proof}

We will deal with a somewhat special subcase first.
\vspace{0.5\baselineskip}

{\bf Case~2A:  } Both of the following properties hold.

\noindent
1. One of $L$, $R$ contains exactly one point of $S$; its index is odd.\\
2. Point $e_1$ is a front point with respect to $v_E$.

\begin{figure}[hbtp]%
\begin{minipage}[t]{0.28\textwidth}
\begin{center}%
\includegraphics[width=\textwidth,keepaspectratio]{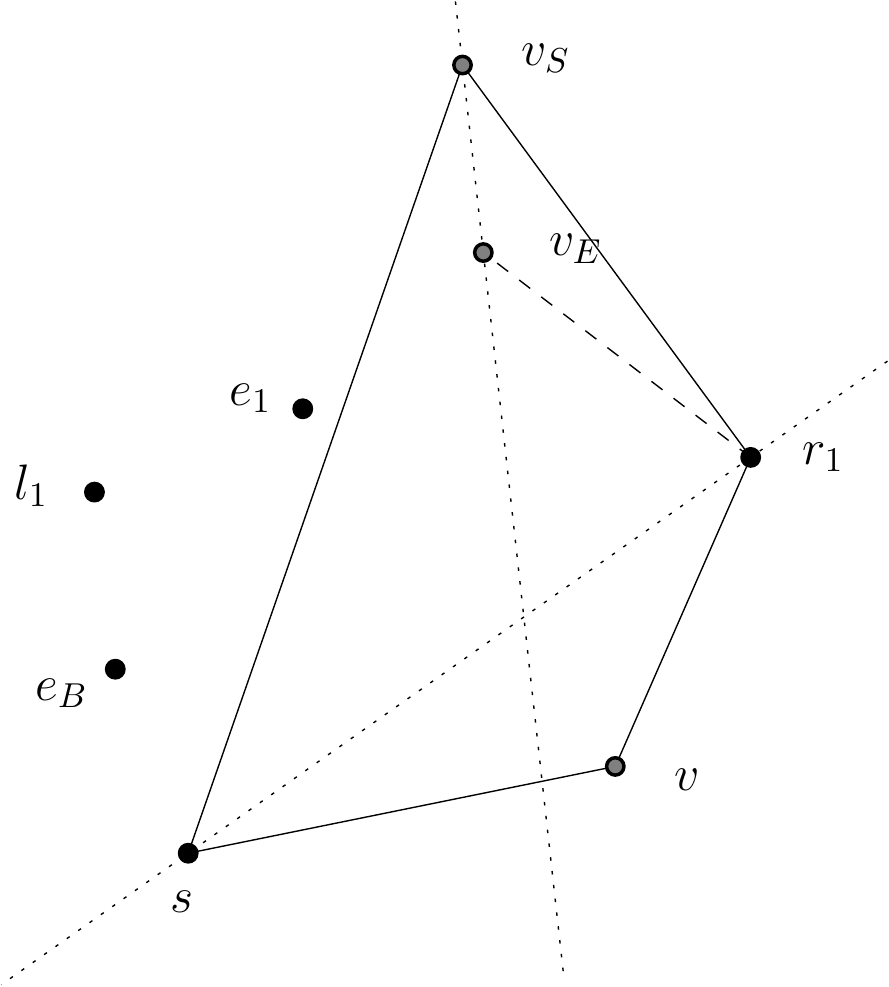}%

(i)
\end{center}%
\end{minipage}
\hfill
\begin{minipage}[t]{0.34\textwidth}
\begin{center}%
\includegraphics[width=\textwidth,keepaspectratio]{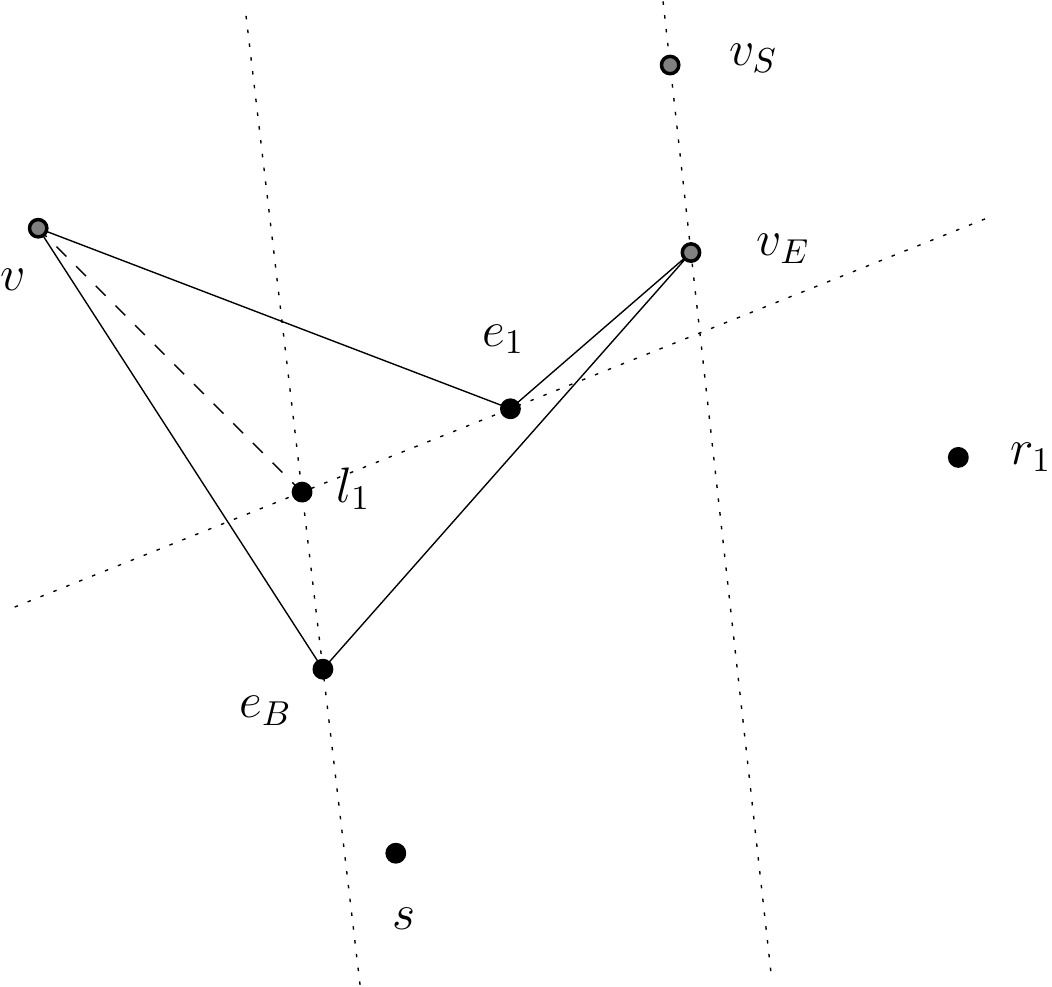}%

(ii)
\end{center}%
\end{minipage}
\hfill
\begin{minipage}[t]{0.34\textwidth}
\begin{center}%
\includegraphics[width=\textwidth,keepaspectratio]{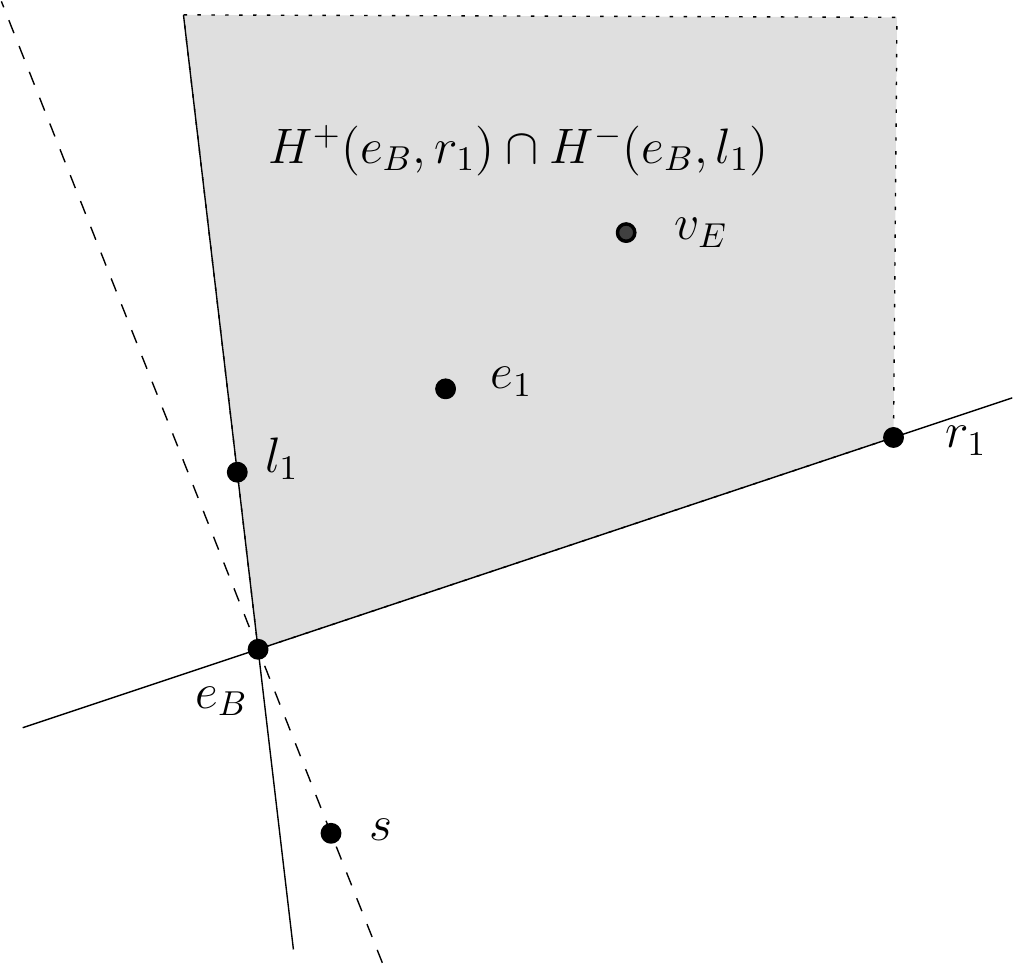}%

(iii)
\end{center}%
\end{minipage}
\caption{The proof of Case 2A.}
\label{NoOddFrontpointsEvenFrontpoint-fig}
\end{figure}

In this case, $v_E$ must lie outside $ch(S)$, and $e_1$ is the only front point with respect to $v_E$, so that $l_1$ and $r_1$ 
are tangent points, as seen from $v_E$.  Moreover, $v_E$ lies in the triangle given by $l_1,r_1$ and $v_S$ 
because otherwise $e_1$ would be a back point.

W.l.o.g., let $r_1$ be the only point in $R$, and let $L$ be situated to the right of the directed line from $v_S$ through $v_E$;
see Figure~\ref{NoOddFrontpointsEvenFrontpoint-fig}.
Then every view point $v$ that sees $r_1$ and some point $s$ of $L$ lies on the same side of $L(r_1,v_E)$ as $v_S$ does,
or $ \overline{r_1v_E}$ would be encircled by $r_1-v-s-v_S-r_1$. 
For the same reason every view point $v$ that sees $r_1$ and some point $s$ of $L$ lies on the same side of $L(s,r_1)$ as $v_S$ does,
see Figure \ref{NoOddFrontpointsEvenFrontpoint-fig} (i).

Also, $r_1$ does not see any other point $s\in S$, otherwise $\overline{r_1v_E}$ would be encircled by $r_1-s-v_S-r_1$.

Now let $e_B$ be the even-indexed neighbour of $l_1$ that is a back point, and let us set 
$Q=\mbox{vis}(e_B)\cap\mbox{vis}(e_1) \cap\mbox{vis}(r_1) \cap \mbox{vis}(l_1)^c$. 

Next, we want to show that every view point $v \in Q$ lies in $H^-(e_B,l_1)$;
see Figure~\ref{NoOddFrontpointsEvenFrontpoint-fig}~(ii). If this were wrong,
$v  \in H^-(r_1,e_B)\cap H^-(r_1,e_1)\subset H^-(r_1,l_1)$ would imply
$v \in H^+(l_1,e_1)$. Since $v_E$ obviously lies in $H^+(l_1,e_1)\cap H^-(e_B,l_1)$,
we could apply Lemma~\ref{tria-lem} to $(a,b,c,w)=(e_1,l_1,e_B,v_E)$, and obtain
$v \in H^-(e_B,l_1)$---a contradiction.

Now let us assume that, in addition to being in $Q$, view point $v$ sees a 
point $s \in S \setminus \{l_1, r_1, e_1, e_B\}$.
As $v$ lies in $H^-(r_1,e_B)$, it follows that $v\in H^-(s,e_B)\supset H^-(r_1,e_B)\cap H^-(e_B,l_1)$, see~(iii). 

On the other hand, $v$ also lies in $H^-(r_1,s)$ as already shown.

Summarizing, we have obtained a result analogous to Lemma~\ref{main-lem}. 

\begin{lemma} \label{case2A-lem}
Let $s \in S \setminus \{l_1, r_1, e_1, e_B\}$. Then each view point in $Q$ that 
sees $s$ lies in the wedge $U_s = H^-(s,e_B) \cap H^-(r_1,s)$.
\end{lemma}%

Now the proof of Case~2A is completed by exactly the same arguments used
subsequently to Lemma~\ref{main-lem} in Section~\ref{outer-sec}.  

\bigskip

If one of the properties of Case~2A is violated, we obtain
the following, by logical negation.
\vspace{0.5\baselineskip}

{\bf Case~2B:  } At least one of the following properties holds.

\noindent
1. None of $L$, $R$ is a singleton set containing an odd indexed point.

\noindent 2. Point $e_1$ is a back point, as seen from $v_E$.

\medskip
Other than in the previous cases, we will now reduce visibility regions to {\em half-planes}, rather than to wedges.
We will show the existence of three points, $p_1, p_2, p_3$ in $S$, and of a halfplane $H_{i}$ 
for each, such that the following holds. 
Let $Q$ denote the set of view points that see at least $S \setminus \{p_1,p_2,p_3 \}$.
Then,
\begin{eqnarray}     \label{3halfplanes}
       \mbox{for each }  v \in Q:  \mbox{ for each } i = 1,2,3: \ \       v  \mbox{ sees } p_i  \Longleftrightarrow     v\in H_{i}.
\end{eqnarray}
Property~\ref{3halfplanes} leads to a contradiction, due to the following analogon of Theorem~\ref{isler-theo}.

\begin{lemma}     \label{isler-lem}
For any arrangement of three (or more) half-planes, there is a subset $T$ of half-planes for which no cell
is contained in exactly the half-planes of $T$.
\end{lemma}
\begin{proof}
With three half-planes, we have eight subsets, but at most seven cells.
\end{proof}
While this fact is easier to prove, and somewhat more efficient, as we need only three points to derive a contradiction, 
it is harder to find points fulfilling Property~\ref{3halfplanes}. This will be our next task.
Again, we consider points in $L$ and points in $R$ separately. Let us discuss the situation for $L$.

We start by defining two points, $l_1'$ and $e'$. 
Suppose there is a point with odd index in $R$. We set $l_1'=l_1$. 
As we are in Case~2B, point $e_1$---situated between $l_1$ and $r_1$---is a back point,
or there is some point $e$ with even index in $R$. 
In the first case we set $e':=e_1$, in the latter case we set $e':=e$; see Figure~\ref{NomenklaturCaseB-fig}~(i) and~(ii).

If there is no point with odd index in $R$ then there are five points with odd index in $L$. We then set $l_1'$ to be the second point with 
odd index that was hit during the rotation of the half-line from $v_E$ through $v_S$. 
Then $l_1$ and $l_1'$  are distinct back points with respect to $v_E$. 
Between $l_1$ and $l_1'$ on the boundary of the convex hull there lies exactly one point $e$ that has even index. We set $e'=e$.
In this case, there are three points with even index on the convex hull between $l_1'$ and $l_2$.
Notice that $e'$ is a back point with respect to $v_E$; see Figure~\ref{NomenklaturCaseB-fig}~(iii).

In either case the points $l_1'$ and $l_2$ have odd indices, 
and the point $e'$ has even index and is either 
a back point with respect to $v_E$, or it lies in $R$.

\begin{figure}[hbtp]%
\begin{minipage}[t]{0.26\textwidth}
\begin{center}%
\includegraphics[width=\textwidth,keepaspectratio]{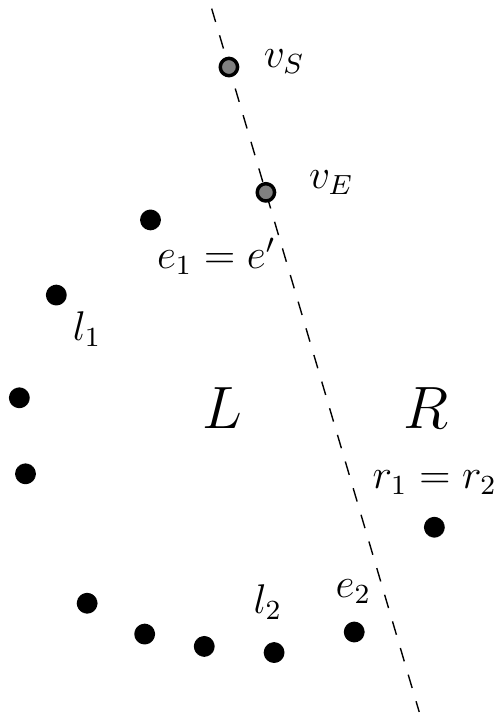}%

(i)
\end{center}%
\end{minipage}
\hfill
\begin{minipage}[t]{0.37\textwidth}
\begin{center}%
\includegraphics[width=\textwidth,keepaspectratio]{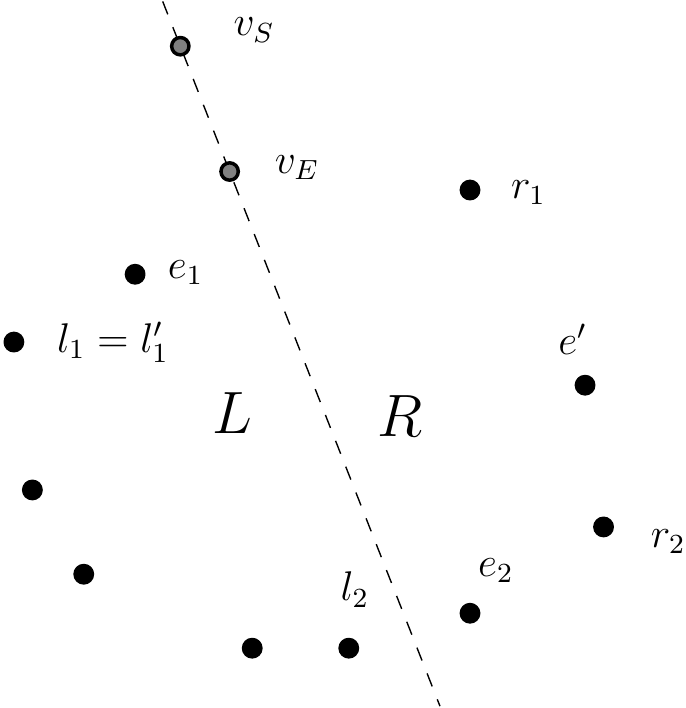}%

(ii)
\end{center}%
\end{minipage}
\hfill
\begin{minipage}[t]{0.28\textwidth}
\begin{center}%
\includegraphics[width=\textwidth,keepaspectratio]{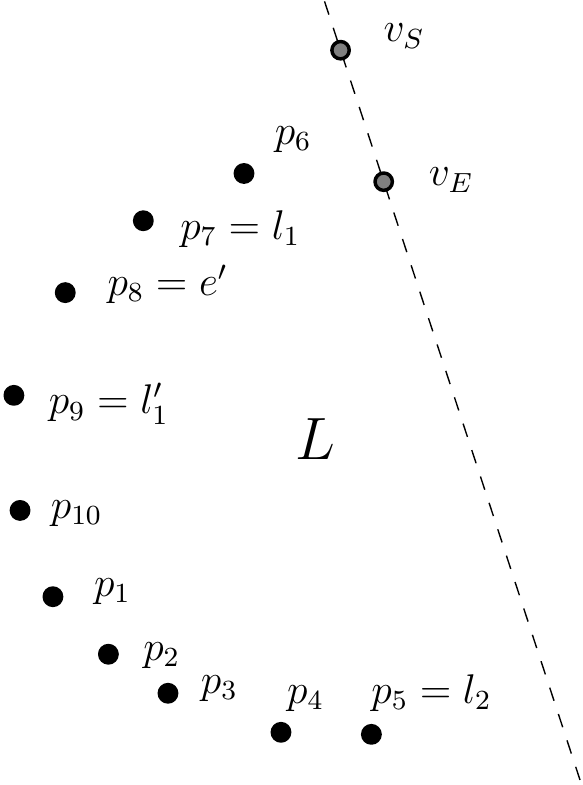}%

(iii)
\end{center}%
\end{minipage}
\caption{(i) If $e_1$ is a back point with respect to $v_E$ we set $e'=e_1$. 
(ii) Otherwise there is an even indexed point in $R$ we will
call $e'$.
(iii) If there is no (odd) point in R then $l_1'$ is the second odd indexed point and $e'$ is the (back) point between $l_1$ and $l_2$.}
\label{NomenklaturCaseB-fig}
\end{figure}

We will now prove the following.

\begin{lemma} \label{existenceHalfplanes-lem}
For all back points $p$ with even index that lie in the wedge given by the rays from $v_E$ through $l_1'$ and $l_2$ the following holds.
There is a half-plane $H_p$ such that every view point $v$ that sees $l_1'$ and $l_2$ sees $p$ if and only if $v\in H_p$.
The analogue holds if we replace $l_1'$ and $l_2$ by $r_1'$ and $r_2$.
\end{lemma}
 
 Before we prove Lemma~\ref{existenceHalfplanes-lem}, we first use it to derive the following consequence. As explained before, it 
 provides us with a contradiction, thus proving Case~2B of Lemma \ref{outer-lem} and completing all proofs.
 
 \begin{lemma}
 There are three points $p_1,p_2,p_3\in S$ and half-planes $H_{1},H_{2},H_{3}$ that satisfy Property~\ref{3halfplanes}.
 \end{lemma}
 \begin{proof}
 If there are odd points in both $L$ and $R$, then there is exactly one even point between $l_1'$ and $r_1'$ and one even point 
 between $l_2$ and $r_2$ and all other even points lie between the rays from $v_E$ through $l_1'$ and $l_2$ and through $r_1'$ and $r_2$, respectively. By Lemma \ref{existenceHalfplanes-lem}, we get that the remaining three even-indexed points have the desired property.
 
 If there is no odd point in $R$ or in $L$, then there are four even-indexed points between $l_1$ and $l_2$ or between $r_1$ and $r_2$ and therefore there are three points with the desired property between $l_1'$ and $l_2$ or between $r_1'$ and $r_2$.
 \end{proof}

\begin{proof}
To prove Lemma \ref{existenceHalfplanes-lem} let $e\in S$ be a point with even index that lies between $l_1'$ and $l_2$. Points $e$ and $v_S$ lie on opposite sides of $L(l_1',v_E)$, by the definition of $l_1'$.

\begin{figure}[hbtp]%
\begin{minipage}[t]{0.30\textwidth}
\begin{center}%
\includegraphics[width=\textwidth,keepaspectratio]{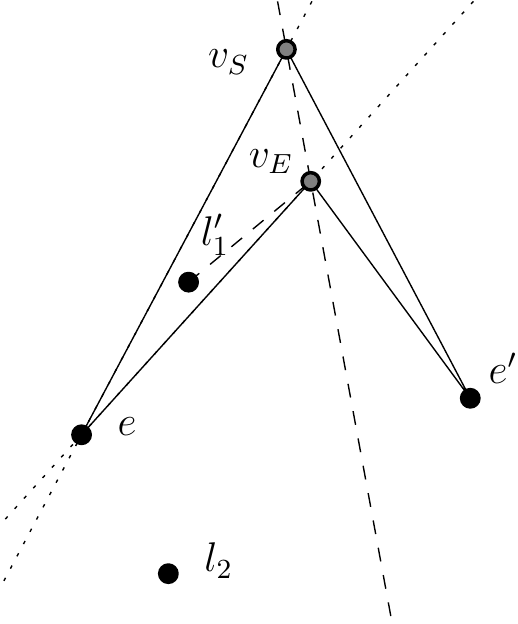}%

(i)
\end{center}%
\end{minipage}
\hfill
\begin{minipage}[t]{0.23\textwidth}
\begin{center}%
\includegraphics[width=\textwidth,keepaspectratio]{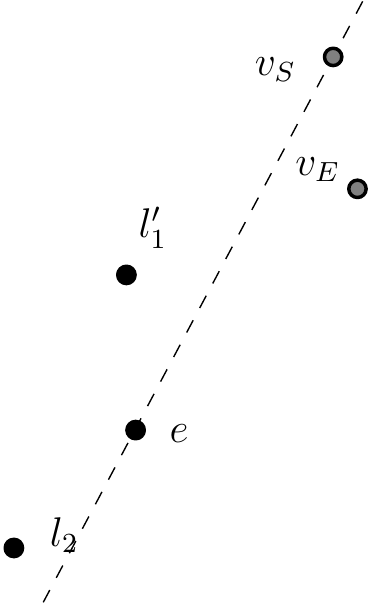}%

(ii)
\end{center}%
\end{minipage}
\hfill
\begin{minipage}[t]{0.28\textwidth}
\begin{center}%
\includegraphics[width=\textwidth,keepaspectratio]{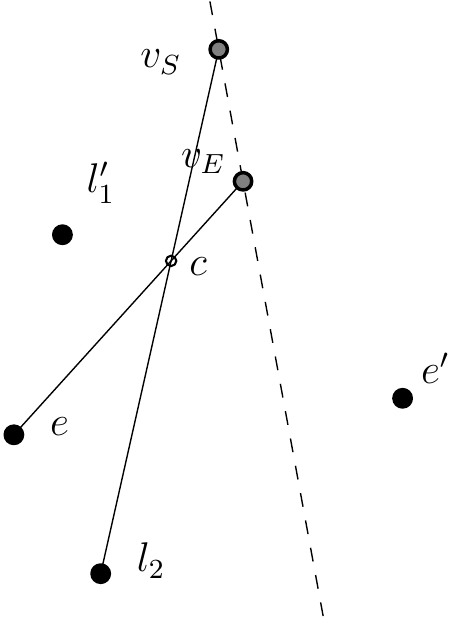}%

(iii)
\end{center}%
\end{minipage}
\caption{(i) $l_1'$ cannot lie between the rays $L(e,v_S)$ and $L(e,v_E)$. 
(ii) $l_1'$ and $l_2$ can not lie on the same side of $L(e,v_S)$.
(iii) So there must be an intersection between $\overline{ev_E}$ and $\overline{l_2v_S}$.}
\label{CaseBIntersectClaim-fig}
\end{figure}

\begin{claim}In this situation the segments $\overline{ev_E}$ and $\overline{l_2v_S}$ intersect in a point $c$. 
\end{claim}

\begin{proof}The segment $\overline{l_2v_S}$ intersects the line $L(e,v_E)$ by definition of $l_2$. It remains to show that $l_2$ does neither lie on the side of $L(v_S,v_E)$ opposite to $e$ nor on the side of $L(v_S,e)$ opposite to $v_E$.
As $l_2$ and $e$ both belong to $L$, the first assertion follows.
For the second one, notice that $l_1'$ cannot lie on the same side of $L(e,v_S)$ as $v_E$ does because otherwise  $\overline{l_1'v_E}$ would be encircled by $e-v_E-e'-v_S$, see Figure \ref{CaseBIntersectClaim-fig} (i). But $l_1'$ and $l_2$ cannot both lie on the side of $L(e,v_S)$ opposite to $v_E$: Because $l_1',l_2,e$ are backpoints, $v_E,l_1',e$ and $l_2$ are the corners of a convex quadrilateral. If $l_1'$ and $l_2$ lie on the same side of a line through $e$, this line must be a tangent to this quadrilateral and therefore $l_1',l_2$ and $v_E$ would have to lie on the same side of this line, see (ii).
So the segment $\overline{ev_E}$ crosses $\overline{l_2v_S}$ in a point $c$.
\end{proof}

\begin{figure}[hbtp]%
\begin{minipage}[t]{0.30\textwidth}
\begin{center}%
\includegraphics[width=\textwidth,keepaspectratio]{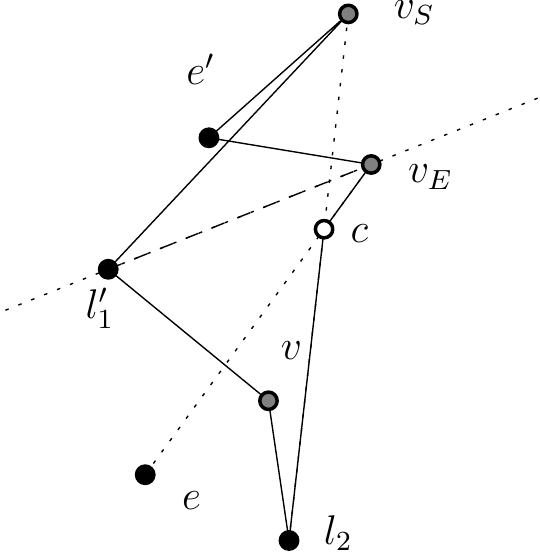}%

(i)
\end{center}%
\end{minipage}
\hfill
\begin{minipage}[t]{0.20\textwidth}
\begin{center}%
\includegraphics[width=\textwidth,keepaspectratio]{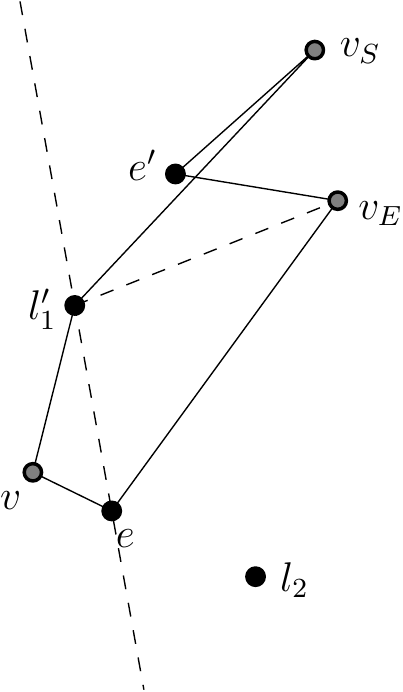}%

(ii)
\end{center}%
\end{minipage}
\hfill
\begin{minipage}[t]{0.28\textwidth}
\begin{center}%
\includegraphics[width=\textwidth,keepaspectratio]{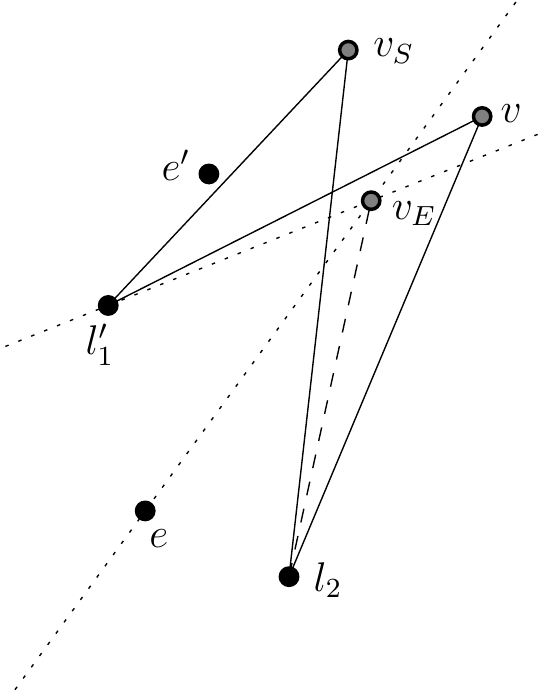}%

(iii)
\end{center}%
\end{minipage}
\caption{ (i)$v$ and $v_S$ must lie on the same side of $L(l_1',v_E)$.  
(ii) $v$ and $v_S$ must lie on the same side of $L(e,l_1')$.
(iii) $v$ and $v_S$ must lie on the same side of $L(e,v_E)$.}
\label{CaseBWedgeLemma-fig}
\end{figure}

Now it follows that every view point $v$ that sees $l_1'$ and $l_2$ lies on the same side of $L(l_1',v_E)$ as $v_S$ does, because otherwise the segment $\overline{l_1'v_E}$ would be encircled by $v_E-c-l_2-v-l_1'-v_S-e'-v_E$, see Figure \ref{CaseBWedgeLemma-fig} (i).

It also follows that every view point that sees $l_1',l_2$ and $e$ has to lie on the same side of $L(e,l_1')$ as $v_S$ does, because otherwise $\overline{l_1'v_E}$ would be encircled by $v_E-e'-v_S-l_1'-v-e-v_E$, see Figure \ref{CaseBWedgeLemma-fig} (ii).

\begin{claim} \label{l_1l_2-claim} Every view point $v$ that sees $l_1'$ and $l_2$ lies on the same side of $L(e,v_E)$ as $v_S$ does. 
\end{claim}

\begin{proof}
Assume $v$ and $v_S$ lay on opposite sides of  $L(e,v_E)$. We already showed that $v$ must lie on the same side of $L(l_1',v_E)$ as $v_S$ does. So $ \overline{l_2v_E}$ would be encircled by $l_2-v-l_1'-v_S-l_2$, see Figure \ref{CaseBWedgeLemma-fig} (iii).
\end{proof}

\begin{lemma}\label{Case1aWedgeSummary} All view points $v$ that see $\{l_1',l_2,e\}$
lie in the wedge $W$ given by the two rays originating in $e$ and going through $l_1'$ and $l_2$, respectively
\end{lemma}

\begin{proof}
We just showed that all such view points $v$ lie in the wedge $W_e$ 
given by the two rays originating in $e$ and going through $v_E$ and $l_1'$, respectively. 
As $W_e$ is a subset of $W$, the lemma follows.
\end{proof}

\begin{figure}[hbtp]%
\begin{minipage}[t]{0.30\textwidth}
\begin{center}%
\includegraphics[width=\textwidth,keepaspectratio]{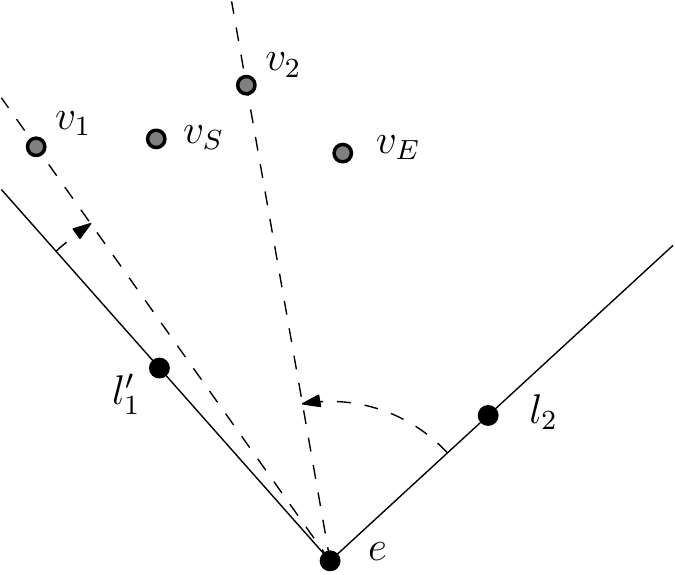}%

(i)
\end{center}%
\end{minipage}
\hfill
\begin{minipage}[t]{0.20\textwidth}
\begin{center}%
\includegraphics[width=\textwidth,keepaspectratio]{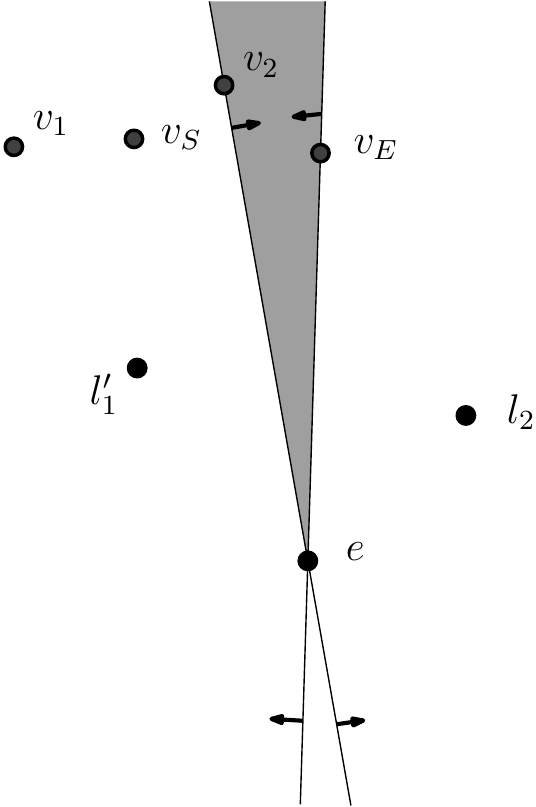}%

(ii)
\end{center}%
\end{minipage}
\hfill
\begin{minipage}[t]{0.28\textwidth}
\begin{center}%
\includegraphics[width=\textwidth,keepaspectratio]{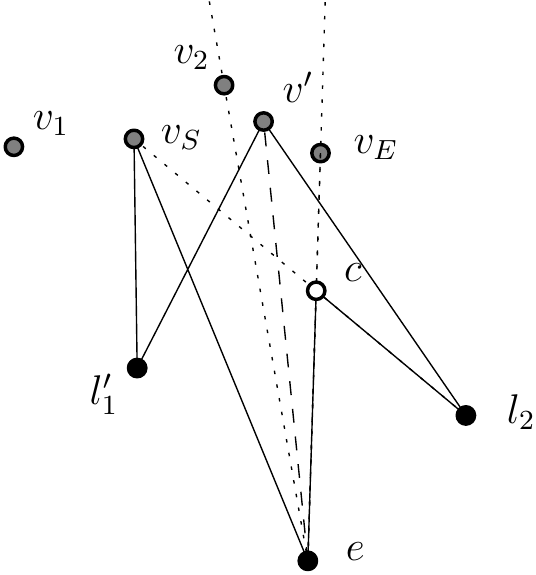}%

(iii)
\end{center}%
\end{minipage}
\caption{(i) We rotate the rays through $l_1'$ and $l_2$ until they encounter $v_1$ and $v_2$.  
(ii) The area between $L(e,v_E)$ and $L(e,v_2)$.
(iii) No point $v'$ that lies in this area sees $l_1'$ and $l_2$ but not $e$.}
\label{CaseBHalfplaneLemma-fig}
\end{figure}

Let us now rotate the ray with origin $e$ through $l_2$  over the wedge $W$, towards $l_1'$. Let us denote the first view point we encounter that sees $l_1', l_2$ and $e$  by $v_2$. Let us then rotate the ray with origin $e$ through $l_1'$  over the wedge $W$, towards $l_2$. Let us denote the first view point we encounter this time that sees $l_1', l_2$ and $e$  by $v_1$, 
see Figure \ref{CaseBHalfplaneLemma-fig} (i).
 
We now obtain the two following facts.

\begin{claim} \label{v_1v_2-claim}
All view points that see $l_1', l_2$ and $e$ lie in the wedge originating in $e$ and going through $v_1$ and $v_2$.
\end{claim}

\begin{proof}
By Lemma \ref{Case1aWedgeSummary} we know that all such points lie between $l_1'$ and $l_2$. By construction of $v_1$ and $v_2$ there is no such point between $l_1$ and $v_1$ or between $l_2$ and $v_2$.
\end{proof}

\begin{lemma}\label{Case1aHalfplaneSummary}There is no view point on the side of $L(e,v_2)$ opposite to $v_S$ that sees $l_1'$ and $l_2$ but not $e$.
\end{lemma}

\begin{proof}
By Claim \ref{l_1l_2-claim}, all view points that see $l_1'$ and $l_2$ lie on the same side of $L(e,v_E)$ as $v_S$ does.  
As $v_2$ also lies on this side of the line and moreover inside the wedge between the rays from $e$ through $v_E$ and $v_S$, respectively, it follows, that a point that sees $l_1'$ and $l_2$ and that lies on the side of $L(e,v_2)$ opposite to $v_S$ must lie in the wedge given by the rays from $e$ through $v_E$ and $v_2$, which in turn is contained in the wedge between the rays from $e$ through $v_E$ and $v_S$, see Figure \ref{CaseBHalfplaneLemma-fig} (ii). 

Now assume there was a view point $v'$ in this wedge, that saw $l_1'$ and $l_2$ but not $e$. If we take $c$ to be the intersection of $\overline{ev_E}$ and $\overline{l_2v_S}$, then the segment $\overline{ev'}$ would be encircled by $e-c-l_2-v'-l_1'-v_S-e$, 
see Figure \ref{CaseBHalfplaneLemma-fig} (iii).

\end{proof}

Now we are able to complete the proof of Lemma \ref{existenceHalfplanes-lem}.

We define $H_e$ to be the closed halfplane to the side of the line through $e$ and $v_1$ in which $v_2$ lies. By Claim \ref{v_1v_2-claim} all view points that see $l_1',l_2$ and $e$ lie in $H_e$. Assume now there was a view point $v$ in $H_e$ 
that sees $l_1'$ and $l_2$ but not $e$. 
By Lemma \ref{Case1aHalfplaneSummary} and the assumption that $v$ lies in $H_e$, it follows that then $v$ must lie in the wedge with origin $e$ and rays through $v_1$ and $v_2$. 

This again leads to a contradiction because the segement $\overline{ev}$ then would be encircled by $e-v_1-l_1'-v-l_2-v_2-e$.
So a view point $v$ that sees $l_1'$ and $l_2$ sees $e$ if and only if $v\in H_e$.
\end{proof}
Now all proofs are complete.

\section{Conclusions}      \label{conclu-sec}                                           

In his classical proof in~\cite{m-ldg-02}, Matou\v sek used
a particular type of enclosing cycle of length~4 to show that the VC-dimension of visibility
regions in simple polygons is finite (obtaining a bound in the thousands). Valtr~\cite{v-ggwps-98}
was able to prove an upper bound of~23 by combining enclosing chain arguments with a
cell decomposition technique. Our proof yields an upper bound of~14, using enclosing
cycles of length~6. The natural question is if better bounds can be obtained by considering
even more complex enclosing configurations, and if there is a systematic way to approach
this problem. One would expect that the true value of the VC-dimension is
closer to 6 than to 14.

\section{Acknowledgement}
The second author would like to thank Boris Aronov and David Kirkpatrick for interesting discussions. We would also like to thank the anonymous referees who carefully read the SoCG '11 version of this paper.

\bigskip

\bibliography{VcProc}{}
\bibliographystyle{plain}

\end{document}